\documentclass[11pt]{article}

\usepackage[utf8]{inputenc}
\usepackage[T1]{fontenc}

\usepackage{amsfonts, 
            amssymb, 
            graphicx, 
            a4, 
            epic, 
            eepic, 
            epsfig, 
            setspace, 
            lscape, 
            bbold}
\usepackage{bbm}
\usepackage{subfigure}
\usepackage{enumitem}
\usepackage{parskip}
\usepackage{amsmath}
\usepackage{xcolor}
\usepackage{adfbullets, adfarrows}
\usepackage{mathtools} 
\usepackage[margin=16pt,font=small,labelfont=bf]{caption} 
\usepackage[pdfencoding=unicode, colorlinks=true]{hyperref} 

\newcommand{\inst}[1]{$^{#1}$}

\DeclarePairedDelimiter{\scProd}{\langle}{\rangle}

\DeclarePairedDelimiter{\braces}{\{}{\}}

\newtheorem{theorem}{Theorem}[section]
\newtheorem{lemma}[theorem]{Lemma}

\newtheorem{definition}[theorem]{Definition}

\def \qedA {\hspace*{\fill} $\square$}
\def \qed {{\square\hfill}}

\newenvironment{proof}{
    \par
    \emph{Proof: }
} {
    \qedA
}

\newcommand{\ZZ}{\mathbb{Z}}

\newcommand{\RR}{\mathbb{R}}

\newcommand{\cX}{{\cal X}}

\newcommand{\cE}{{\cal E}}

\newcommand{\numpluses}[1] {\ensuremath{|V_{+}(#1)|}}

\newcommand{\spinclassA} {\ensuremath{\Lambda\negthinspace{\genfrac{}{}{0pt}{2}{--}{\phantom{-}-}}}}
\newcommand{\spinclassB} {\ensuremath{\Lambda\negthinspace{\genfrac{}{}{0pt}{2}{+-}{\phantom{-}-}}}}
\newcommand{\spinclassC} {\ensuremath{\Lambda\negthinspace{\genfrac{}{}{0pt}{2}{+-}{\phantom{-}+}}}}
\newcommand{\spinclassD} {\ensuremath{\Lambda\negthinspace{\genfrac{}{}{0pt}{2}{++}{\phantom{-}+}}}}
\newcommand{\spinclassE} {\ensuremath{\Lambda\negthinspace{\genfrac{}{}{0pt}{2}{-+}{\phantom{-}+}}}}
\newcommand{\spinclassF} {\ensuremath{\Lambda\negthinspace{\genfrac{}{}{0pt}{2}{-+}{\phantom{-}-}}}}

\newcommand{\numspinsB} {\left|\spinclassB\right|}
\newcommand{\numspinsC} {\left|\spinclassC\right|}
\newcommand{\numspinsD} {\left|\spinclassD\right|}
\newcommand{\numspinsE} {\left|\spinclassE\right|}
\newcommand{\numspinsF} {\left|\spinclassF\right|}

\newcommand{\indicator}[1]{\ensuremath{\chi_{\braces{#1}}}}

\begin{document}

\title{Shaken dynamics: an easy way to parallel Markov Chain Monte Carlo}

\author{
Valentina Apollonio\inst{1}\and
Roberto D'Autilia\inst{1}\and
Benedetto Scoppola\inst{2} \and
Elisabetta Scoppola\inst{1}\and
Alessio Troiani\inst{3,4}}

\maketitle
 
\begin{center} {
    \footnotesize
    \vspace{0.3cm} \inst{1}  Dipartimento di Matematica e Fisica, Universit\`a
    Roma Tre\\
    Largo San Murialdo, 1 - 00146 Roma, Italy\\

    \vspace{0.3cm}\inst{2}  Dipartimento di Matematica, Universit\`a di Roma
    ``Tor Vergata''\\
    Via della Ricerca Scientifica, 1 - 00133 Roma, Italy\\

    \vspace{0.3cm}\inst{3}  Dipartimento di Matematica ``Tullio Levi--Civita'', Universit\`a di Padova\\
    Via Trieste, 63 - 35121 Padova, Italy

    \vspace{0.3cm}\inst{4}
    Dipartimento di Matematica ``Guido Castelnuovo'', Università
    di Roma ``La Sapienza''\\
    Piazzale Aldo Moro, 5 - 00185 Roma, Italy
}

\end{center}

\vskip1.truecm

\begin{abstract}
    \noindent
    We define a class of Markovian parallel dynamics for spin systems on arbitrary graphs with nearest neighbor interaction described by a Hamiltonian function $H(\sigma)$.\\
    These dynamics turn out to be reversible and their stationary measure is explicitly determined.\\
    Convergence to equilibrium and relation of the stationary
    measure to the usual Gibbs measure are discussed when the 
    dynamics is defined on $\mathbb{Z}^2$.\\
    Further it is shown how these dynamics can be used to define natively parallel algorithms to face problems in the context of combinatorial optimization.
\end{abstract}

\section{Introduction}
\label{Intro}

We introduce a class of parallel dynamics,
called \emph{shaken dynamics}, to study spin systems on arbitrary graphs $G=(V, E)$ with general interaction given by
\begin{align}\label{eq:reference_hamiltonian}
H(\sigma)=-\sum_{e=\{x, y\} \in E} J_{x y} \sigma_{x} \sigma_{y}-2 \sum_{x \in V} \lambda_{x} \sigma_{x}
\end{align}
with $J_{x y}$ and $\lambda_{x}$ in $\mathbb{R}$ for 
$x, y \in V$, and $\sigma \in\{-1,+1\}^{V}$ a configuration on $G$.
The novelty of these dynamics is that transitions between configurations are obtained
through a combination of pairs of \emph{half steps}
each characterized by an \emph{asymmetric} interaction.

The shaken dynamics introduced here extend a class of PCA on spin systems, appeared in previous papers \cite{dss1, dss2, LS, pss1}, characterized by transition probabilities defined in terms of a pair Hamiltonian, i.e.

\begin{equation*}
    P(\sigma,\tau)=\frac{e^{-\beta H(\sigma,\tau)}}{\sum_{\tau'}e^{-\beta H(\sigma,\tau')}}
\end{equation*}

with $H(\sigma,\tau)=\sum_xh_x(\sigma)\tau_x$ and $\beta=\frac{1}{T}$, where $T$ is the temperature of the system. For each vertex $x$ the local field $h_x(\sigma)$ depends on the value of the spins in a neighborhood of $x$ and on the value of the spin at site $x$ itself through a self-interaction parameter $q>0$.
The dynamics considered in the aforementioned papers turn out to be reversible when the interaction with the spins in the neighborhood is symmetric and
irreversible otherwise. On one hand, reversibility allows for a better control of the stationary measure.
On the other hand, irreversible dynamics may exhibit a faster convergence to equilibrium
(see, e.g., \cite{KJZ, dss2}), though the control of the invariant measure is in general more complicated.

For the shaken dynamics we show that reversibility holds despite the presence of a non symmetric interaction. This fact allows for a
rather robust control of the invariant measure on arbitrary graphs also in the case of non-zero external field and different choices of boundary conditions.

We study extensively the shaken dynamics on the square lattice. In this case
we show that, for suitable values of the parameters, its stationary measure
tends, in the thermodynamic limit, to the Gibbs measure with Hamiltonian \eqref{eq:reference_hamiltonian}. 
This result is complemented by a thorough numerical investigation
carried over in \cite{d2021parallel}.
Moreover, we analyze the convergence
to equilibrium of the shaken dynamics in the low-temperature regime. We show that the asymmetric interaction induces a faster convergence with respect to symmetric PCA and single spin flip dynamics. In this respect, it is apparent that
the shaken dynamics retains some of the properties of the irreversible PCA considered
in \cite{dss2, pss2, LS}.
Hence, shaken dynamics benefits of the advantages 
of both reversible and irreversible dynamics.

A notable feature of shaken dynamics is the fact that their invariant measure turns out to be the marginal of the Gibbs measure on an induced
bipartite graph $G^b$ whose geometry is related to that of the original 
graph $G$. Tuning the self-interaction parameter $q$ tunes the geometry
of $G^b$. Leveraging on this feature we show how it is possible to control the critical behavior of the stationary measure of the shaken dynamics
defined on the triangular lattice, extending, in this way, the analogous
result provided in \cite{apollonio2019criticality} for the square lattice.

It is worth mentioning that one of the main reasons of interest on PCA,
and, in general, on parallel dynamics, is related to their numerical applications.
Indeed, parallelization could, at least in principle, speed up Markov Chain Monte Carlo. Even though
until a few years ago parallel computing was expensive and tricky, we have now powerful
and cheap parallel architectures, for instance based on GPU or even FPGA.
In this regard, we show how the class of shaken
dynamics defined on general graphs provides a straightforward manner to define \emph{natively} parallel Monte Carlo algorithms that can be used to tackle discrete optimization problems. Algorithms exploiting shaken dynamics are not bound to any particular computing architecture or graph structure and, hence, their performances are likely to benefit from the development of parallel computing often driven by applications not necessarily linked to academic research.

Shaken dynamics on general graphs are defined in
Section~\ref{Ggen}. It is shown how a shaken dynamics on an arbitrary $G=(V, E)$ is related to a dynamics on a corresponding bipartite graph $G^b=(V^b, E^b)$ where spins in each partition are alternatively updated.
Theorem~\ref{t0} identifies the stationary measure of the shaken dynamics and shows that it is the marginal of the Gibbs measure on the configuration space $\{-1, +1\}^{V^b}$. 
Section~\ref{Z2} is devoted to studying the shaken dynamics on $\mathbb{Z}^2$ with homogeneous
interaction. It is shown that, if the temperature of the system is sufficiently small, the invariant measure is close to the Gibbs measure for the Ising model on the square lattice.
As far as convergence to equilibrium is concerned, the Section provides, in the low temperature regime, a comparison for the tunneling times from the metastable to the stable state between the shaken dynamics, a symmetric PCA and a single spin flip dynamics.
Here theoretical results
are complemented by numerical simulations.
The result concerning the critical behavior for the shaken dynamics on the triangular lattice is provided in Section~\ref{geo_disc}.
In Section~\ref{gensh} the definition of shaken dynamics is generalized and Section~\ref{sec:application_to_optimization_problems} shows how these dynamics can be used in the context of combinatorial optimization.
Proofs are provided in Section~\ref{sec:proofs}.
Finally, Section~\ref{dop}
is devoted to final remarks including the possible applications of the shaken dynamics to model tidal dissipative effects in planetary systems.

\section{Shaken dynamics on general graphs}
\label{Ggen}
Let $G=(V,E)$ be a finite weighted graph and $\cX_V = \{-1,1\}^V$ be the set of spin configurations on $V$. We  consider
the nearest neighbor interaction between spins given by the Ising Hamiltonian in the general form:
\begin{align}\label{ham1}
    H(\sigma) 
    & =-\sum_{e=\{x,y\}\in E}J_{xy}\sigma_x\sigma_y -2\sum_{x\in V}\lambda_x\sigma_x\\
    & = -\sum_x\sum_y \frac{1}{2}J_{xy}\mathbb{1}_{\{x,y\}\in E} \sigma_x\sigma_y -2\sum_{x\in V}\lambda_x\sigma_x=-\scProd{\frac{1}{2}{\cal J}\sigma+2\lambda,\sigma}\notag
\end{align}

where the weight $J_{xy} \in \RR$ associated to the edge $\{x,y\}$, represents the interaction,
and can be written in  compact form as a symmetric matrix ${\cal J}$ 
and we denote by $\scProd{\cdot,\cdot} $ the scalar product. The vector $\lambda=\{\lambda_x\}_{x\in V}$ is an external
field, possibly non constant.

We introduce a class of bipartite weighted graphs $G^b=(V^b,E^b)$ {\it doubling the interaction graph } $G$.
The idea is to duplicate the vertex set into two identical copies, $V^{(1)}$ and $V^{(2)}$, representing  the two parts of the 
vertex set of the bipartite graph. For each $x\in V$ we denote
by  $x^{(1)}, x^{(2)}$ the vertices corresponding to $x\in V$ in $V^{(1)}$ and in  $V^{(2)}$ respectively. The edges between 
$x^{(1)}$ and $x^{(2)}$ are all present, for any $x\in V$.  
On the other hand the edges between $x^{(1)}$ and $y^{(2)}$, 
with $x\not= y$, or between $y^{(1)}$ and $x^{(2)}$,
can be present only if $\{x,y\}\in E$. 
Exactly one edge among the two possibilities   
$(x^{(1)},y^{(2)})$ and $(y^{(1)},x^{(2)})$ is
in $E^b$ if $\{x,y\} \in E$.
This means that for any graph $G$ there are many {\it doubling graphs} $G^b$. 
Note that similar doubling graphs have already
been introduced in literature for
different purposes (see \cite{KV}).
More precisely:

\begin{definition}\label{compg}
    A bipartite weighed graph $G^b=(V^b, E^b)$ is the {\bf doubling graph} of $G=(V,E)$ if
    \begin{itemize}[label=\adfbullet{43}]
        \item the vertex set $V^b=V^{(1)}\cup V^{(2)}$ 
                 where the two parts $V^{(1)}$ and $V^{(2)}$ are two identical copies of $V$;
        \item for any $x\in V$ the edge 
                 $(x^{(1)}, x^{(2)})\in E^b$ and we call it a {\bf self-interaction edge};
        \item if $\{x,y\} \in E$ then one, 
                 and only one, between the two edges $\{x^{(1)}, y^{(2)}\}$ and $\{y^{(1)}, x^{(2)} \}$ is in $E^{b}$. We call this kind of edge an {\bf interaction edge}.
    \end{itemize}
\end{definition}

To construct a doubling graph 
starting from the interaction graph $G=(V,E)$, 
define a new oriented graph $G^o=(V, E^o)$ simply
orienting the edges in an arbitrary way. Using the oriented edges the set 
$E^b$ is constructed as follows. 
For any $x\in V$ we have  the self-interaction edge $(x^{(1)}, x^{(2)})\in E^b$  with weight $w(x^{(1)}, x^{(2)})=q$ and for $x\not=y\in V$
we have $(x^{(1)}, y^{(2)})\in E^b$ if and only if  $(x,y)\in E^o$ with weight $w(x^{(1)}, y^{(2)})=J_{xy}$.

Note that the edges in $E^b$ are not oriented. However, by construction, the graph is bipartite, so that for any $e=\{x,y\}\in E^b$ we have $x\in V^{(1)}, y\in V^{(2)}$ or viceversa and so we consider in the definition the natural order in the edges
in $E^b$ by setting $e=(e^{(1)}, e^{(2)})$ with 
$e^{(1)}\in V^{(1)}, e^{(2)}\in V^{(2)}$. 
For this reason we can use the oriented edges in $E^o$ 
in order to define $E^b$. 

We will sometimes omit the superscripts $^{(1)}$ and $^{(2)}$ and we will always consider $(x,y)$ the ordered pair with
$x\in V^{(1)}, \; y\in V^{(2)}$, and $\{x,y\}$ the unordered pair with $x,y\in V^b$.

\begin{figure}
	\centering
 		\subfigure[]{\includegraphics[width=0.4\textwidth]{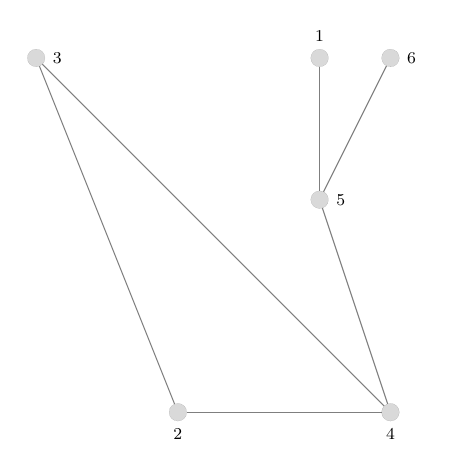}}
		\hfill
		\subfigure[]{\includegraphics[width=0.4\textwidth]{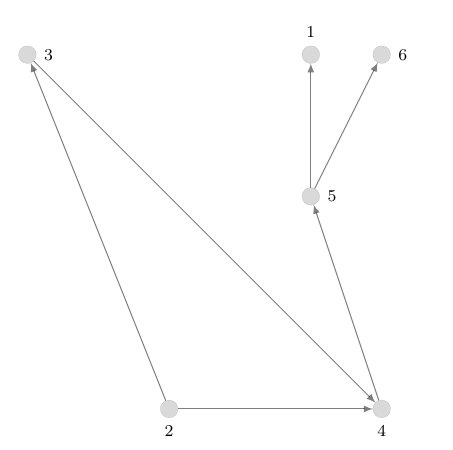}}
		\caption{An undirected graph (a) and a possible choice for the related directed graph (b)}
	\label{fig:undirected_and_directed_graph}
\end{figure}

\begin{figure}[tbh]
\centering
\includegraphics[width=0.7\linewidth]{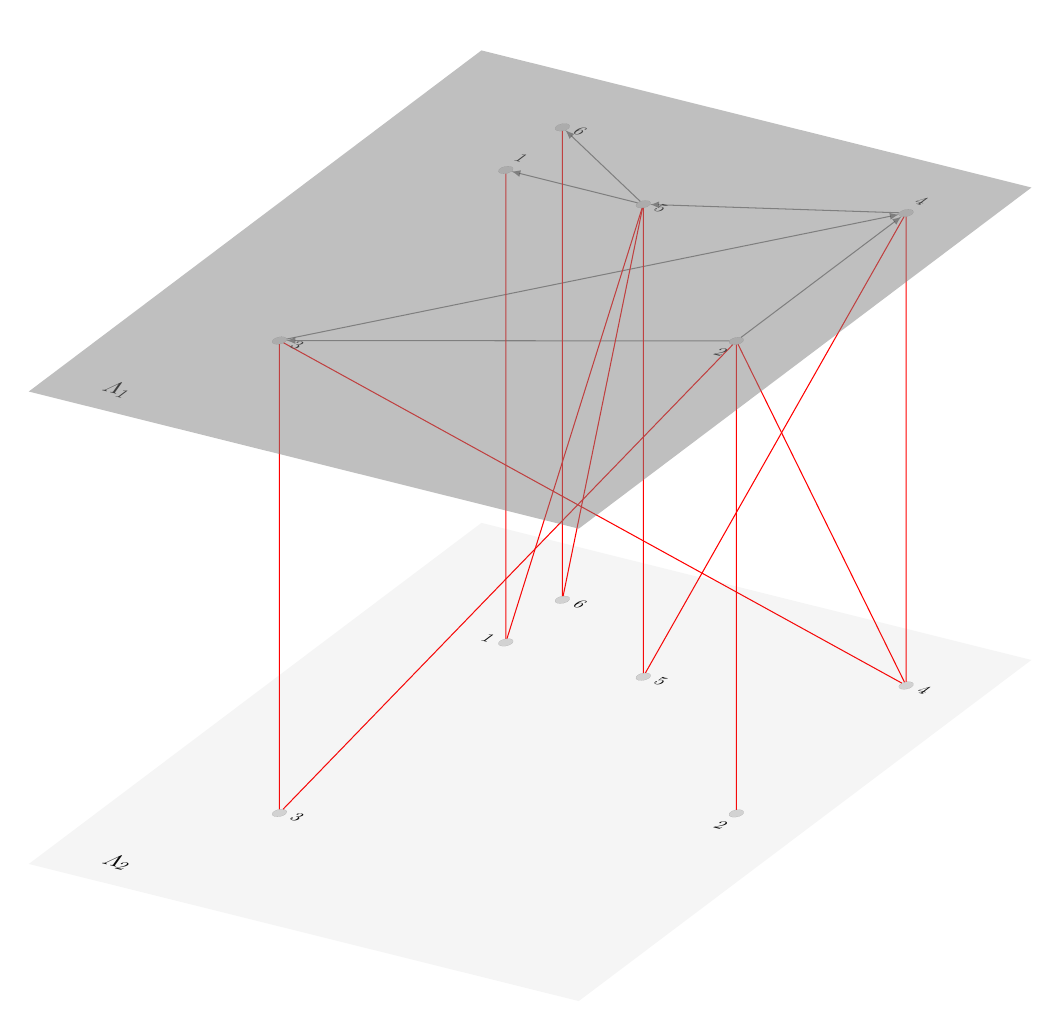}
\caption{The doubling of the graph of 
    Fig.~\ref{fig:undirected_and_directed_graph}(a)
	obtained from the directed graph of Fig.~\ref{fig:undirected_and_directed_graph}(b)}
\end{figure}
\begin{definition}\label{comph}
The pair Hamiltonian $H(\sigma^{(1)},\sigma^{(2)})$ is the {\bf doubling} of the Hamiltonian (\ref{ham1}) with interaction graph $G$ if there exists a doubling graph $G^b=(V^b,E^b)$ of $G$ such that
  $H(\hbox{{\boldmath $ \sigma$}})$,  defined on  the spin configurations $\hbox{{\boldmath $ \sigma$}}\equiv(\sigma^{(1)},\sigma^{(2)})\in \cX_{V^b} = \{-1,1\}^{V^b}$, 
  can be written as 
\begin{align}\label{Hs}
H(\hbox{{\boldmath $ \sigma$}})=-\sum_{\{x,y\}\in E^b}w(x,y)\hbox{{\boldmath $ \sigma$}}_x\hbox{{\boldmath $ \sigma$}}_y
-\sum_{x\in V^b}\lambda_x\hbox{{\boldmath $ \sigma$}}_x
\end{align}
with $w(x,y)=q > 0$ if $\{x,y\}$ is a self interaction edge and  $w(x,y)=J_{xy} \in \mathbb{R}$ otherwise and
with  $\lambda_{x^{(1)}}=\lambda_{x^{(2)}}=\lambda_x \in \mathbb{R}$.
\end{definition}

In a more explicit way we can write
\begin{align}\label{ham21}
\begin{aligned}
H(\hbox{{\boldmath $ \sigma$}}) & \equiv H(\sigma^{(1)},\sigma^{(2)})\\
 	& =-\sum_{(x,y)\in E^b}J_{xy}\sigma^{(1)}_x\sigma^{(2)}_y -\sum_{x\in V}\Big(q\sigma^{(1)}_x\cdot\sigma^{(2)}_x+\lambda_x(\sigma^{(1)}_x+\sigma^{(2)}_x)\Big)\\
	& = -\sum_{x\in V} \Big(\sigma^{(1)}_x h_x^{2\to 1}(\sigma^{(2)})+\lambda_x\sigma^{(2)}_x\Big)\\
	& = -\sum_{x\in V} \Big(\sigma^{(2)}_x h_x^{1\to 2}(\sigma^{(1)})+\lambda_x\sigma^{(1)}_x\Big)
\end{aligned}
\end{align}
with
\begin{equation*}
    h_x^{2\to 1}(\sigma^{(2)}) =
    \sum_{y\in V: (x,y) \in E^b}\Big(J_{xy}\sigma^{(2)}_{y}\Big) +q\sigma^{(2)}_x +\lambda_x
\end{equation*}
and
\begin{equation*}
    h_x^{1\to 2}(\sigma^{(1)}) = 
    \sum_{y\in V: (y,x) \in E^b}
    \Big(J_{xy}\sigma^{(1)}_{y}\Big) +q\sigma^{(1)}_x +\lambda_x
\end{equation*}
By defining ${{\cal J}^o}$ the matrix of oriented interaction, i.e., 
${{\cal J}^o}_{xy}= J_{xy}\mathbb{1}_{(x,y)\in E^o}$, and its transposed ${{\cal J}^o}^T$ corresponding to the opposite orientation,
we can write 
\begin{eqnarray*}
    h_x^{2\to 1}(\sigma^{(2)}) & = & ({{\cal J}^o}\sigma^{(2)})_x+q\sigma^{(2)}_x +\lambda_x
    \\
    h_x^{1\to 2}(\sigma^{(1)}) & = & ({{\cal J}^o}^T\sigma^{(1)})_x+q\sigma^{(1)}_x +\lambda_x
\end{eqnarray*}
and
\begin{align*}
\begin{aligned}
    H(\sigma^{(1)},\sigma^{(2)}) 
    & =
    -\scProd{\sigma^{(1)}, {{\cal J}^o}\sigma^{(2)}}
    -q\scProd{\sigma^{(1)},\sigma^{(2)}}
    -\scProd{\lambda,\sigma^{(1)}}
    -\scProd{\lambda,\sigma^{(2)}} \\
    & =
    -\scProd{{{\cal J}^o}^T\sigma^{(1)}, \sigma^{(2)}}
    -q\scProd{\sigma^{(1)},\sigma^{(2)}}
    -\scProd{\lambda,\sigma^{(1)}}
    -\scProd{\lambda,\sigma^{(2)}}
\end{aligned}
\end{align*}

If we consider the case $\sigma^{(1)}=\sigma^{(2)}=\sigma$, i.e.,  
$\sigma^{(1)}_x=\sigma^{(2)}_x$ for any $x\in V$, then we have
$H(\hbox{{\boldmath $ \sigma$}}) \equiv H(\sigma,\sigma)=H(\sigma)- q|V|$.
Indeed we have immediately ${\cal J}={{\cal J}^o}+{{\cal J}^o}^T$.

We construct now  the {\it shaken dynamics} on the state space $\cX_V$   by considering two subsequent updating
defined as follows:
\begin{equation}
    P^{1\to 2}(\sigma,\sigma'):=\frac{e^{-\beta H(\sigma,\sigma')}}{\overrightarrow{Z}_\sigma}=\prod_{x\in V}\frac{e^{\beta h^{1\to 2}_x(\sigma)\sigma'_x}}{2\cosh{\beta h^{{1\to 2}}_x(\sigma) }}\quad \hbox{ with } \quad \overrightarrow{Z}_\sigma=\sum_{\zeta\in\cX_V}e^{-\beta H(\sigma,\zeta)}
\end{equation}
\begin{equation}
    P^{2\to 1}(\sigma',\tau):=\frac{e^{-\beta H(\tau,\sigma')}}{\overleftarrow{Z}_{\sigma'}}=\prod_{x\in V}\frac{e^{\beta h^{2\to 1}_x(\tau)\sigma'_x}}{2\cosh{\beta h^{{2\to 1}}_x(\tau) }}\quad \hbox{ with } \quad \overleftarrow{Z}_{\sigma'}=\sum_{\zeta\in\cX_V}e^{-\beta H(\zeta,\sigma')}
\end{equation}
\begin{equation}
    P^{sh}(\sigma,\tau)=\sum_{\sigma'\in\cX_V}P^{1\to 2}(\sigma,\sigma')P^{2\to 1}(\sigma',\tau)=\sum_{\sigma'\in\cX_V}
    \frac{e^{-\beta H(\sigma,\sigma')}}{\overrightarrow{Z}_\sigma}\frac{e^{-\beta H(\tau,\sigma')}}{\overleftarrow{Z}_{\sigma'}}
\end{equation}

We state the result on the shaken dynamics in this general context.

\begin{theorem}\label{t0}
The stationary measure  of the shaken dynamics is
\begin{align}\label{mispi}
\pi(\sigma)=\frac{\overrightarrow{Z}_\sigma}{Z} \quad\text{with} \quad \overrightarrow{Z}_\sigma:=\sum_\tau e^{-\beta H(\sigma,\tau)} \; \text{ and }\;
Z:=\sum_{\sigma, \tau }e^{-\beta  H(\sigma,\tau)}
\end{align}
and  reversibility holds. This stationary measure  is the marginal of the Gibbs measure on the space $\cX_{V^b}$ of pairs of configurations
$\hbox{{\boldmath $ \sigma$}}:=(\sigma^{(1)},\sigma^{(2)}) $ defined by:
\begin{equation}\label{pi2}
{{\pi^b}}(\hbox{\boldmath $ \sigma$}):=\frac{1}{Z}e^{-\beta H(\hbox{{\boldmath $ \sigma$}})}.
\end{equation}
The shaken dynamics on $\cX_{V}$ corresponds to an alternate dynamics on $G^b$ in the following sense
\begin{equation}\label{eqshalt}
P^{sh}(\sigma^{(1)},\tau^{(1)})=\sum_{\tau^{(2)}\in \{-1,+1\}^{V^{(2)}}}P^{alt}(\hbox{{\boldmath $ \sigma$},{\boldmath $ \tau$}})
\end{equation}
with
\begin{equation}\label{Palt}
P^{alt}(\hbox{{\boldmath $ \sigma$},{\boldmath $ \tau$}})=
\frac{e^{-\beta H(\sigma^{(1)},\tau^{(2)})}}{\overrightarrow{Z}_{\sigma^{(1)}}}\frac{e^{-\beta H(\tau^{(1)},\tau^{(2)})}}{\overleftarrow{Z}_{\tau^{(2)}}}
\end{equation}
 the stationary measure of $P^{alt}$ is ${{\pi^b}}(\hbox{\boldmath $ \sigma$})$.
 Note that $P^{alt}$ is independent of $\sigma^{(2)}$.
This dynamics is in general non reversible.
\end{theorem}

\section{The shaken dynamics on \texorpdfstring{$\mathbb{Z}^2$}{Z\^2}}
\label{Z2}

Let $\Lambda$ be a two-dimensional $L\times L$ square lattice in $\mathbb{Z}^2$ and
let ${\cal B}_\Lambda$  denote the set of all nearest neighbors
in $\Lambda$ with periodic boundary conditions.

In $\Lambda$ we identify a set $B$ where the value of the spins is frozen throughout the evolution and that plays the role of boundary conditions. This means that we will
consider the state space $\cX_{\Lambda,B}=\{\sigma\in\cX_\Lambda:\; \sigma_x=+1\quad \forall x\in B\}$.

\begin{figure}
    \centering{
        \includegraphics[width=0.45\linewidth]{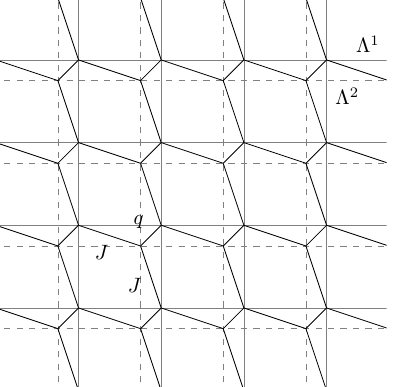}
    }
    \caption{The doubling graph of $\ZZ^{2}$ 
        represented in the figure turns out 
        to be a hexagonal lattice.}
    \label{fig:square_shaken}
\end{figure}

Following the construction of the shaken dynamics of the previous section we can define
\begin{align}\label{ham2}
\begin{aligned}
H(\sigma,\tau) & =
 - \sum_{x \in \Lambda} \left[J \sigma_x(\tau_{x^{\uparrow}} + \tau_{x^{\rightarrow}}) +q\sigma_x \tau_x+ \lambda(\sigma_x+\tau_x)\right]   \\
 & = 
 - \sum_{x \in \Lambda}\left[J \tau_x(\sigma_{x^{\downarrow}} + \sigma_{x^{\leftarrow}})  +q\tau_x\sigma_x + \lambda(\sigma_x+\tau_x)\right]
\end{aligned}
\end{align}
where $x^{\uparrow}, x^{\rightarrow}, x^{\downarrow}, x^{\leftarrow}$ are, respectively, the up, right, down, left neighbors of the site $x$ on the torus $(\Lambda,{\cal B}_\Lambda)$, $J>0$ is the ferromagnetic interaction,  $q>0$ is the inertial constant and $\lambda$ represents the  external field.
We can write
\begin{align}\label{hamasym}
H(\sigma,\tau)=
 - \sum_{x \in \Lambda}\sigma_x h^{ur}(\tau) - \lambda  \sum_{x \in \Lambda}\tau_x=- \sum_{x \in \Lambda}\tau_x h^{dl}(\sigma) - \lambda  \sum_{x \in \Lambda}\sigma_x
\end{align}
  where the local up-right field  $h^{ur}_x(\tau)$ due to the configuration $\tau$ is given by
\begin{equation}\label{hur}
h^{ur}_x(\tau)=\Big[J(\tau_{x^{\uparrow}} + \tau_{x^{\rightarrow}}) +q\tau_x + \lambda\Big]
\end{equation}
and  the local down-left field  $h^{dl}_x(\sigma)$ due to the configuration $\sigma$ is given by
\begin{equation}\label{hdl}
h^{dl}_x(\sigma)=\Big[J(\sigma_{x^{\downarrow}} + \sigma_{x^{\leftarrow}})  +q\sigma_x + \lambda\Big]
\end{equation}
Define the asymmetric updating rule
\begin{equation}\label{pasymm}
P^{dl}(\sigma,\tau):=\frac{e^{-\beta H(\sigma,\tau)}}{\overrightarrow{Z}_\sigma}\quad \hbox{ with } \quad \overrightarrow{Z}_\sigma=\sum_{\sigma'\in\cX_{\Lambda,B}}e^{-\beta H(\sigma,\sigma')}
\end{equation}
Due to the definition of the pair Hamiltonian, the updating performed by the
transition probability $P^{dl}(\sigma,\tau)$ is parallel: given a configuration $\sigma$,
at each site $x\in \Lambda$ the spin $\tau_x$ of the new configuration $\tau$ is chosen with a probability proportional
to $e^{\beta h^{dl}_x(\sigma)\tau_x}$ so that
$$
P^{dl}(\sigma,\tau):=\frac{e^{-\beta H(\sigma,\tau)}}{\overrightarrow{Z}_\sigma}=\prod_{x\in\Lambda}\frac{e^{\beta h^{dl}_x(\sigma)\tau_x}}{2\cosh{\beta h^{dl}_x(\sigma) }}
$$
We have
$H(\sigma,\tau)\not=H(\tau,\sigma)$ and actually, by (\ref{hamasym}), $H(\tau,\sigma)$ corresponds to the opposite direction of the interaction for the
transition from $\sigma$ to $\tau$. We define
\begin{equation}\label{pasymm1}
 P^{ur}(\sigma,\tau):=\frac{e^{-\beta H(\tau,\sigma)}}{\overleftarrow{Z}_{\sigma}}\quad\hbox{ with }\quad  \overleftarrow{Z}_\sigma=\sum_{\sigma'\in\cX_{\Lambda,B}}e^{-\beta H(\sigma',\sigma)}
\end{equation}

Similarly for $P^{ur}(\sigma,\tau)$ with the up-right field $h^{ur}_x(\sigma)$
we get
$$
P^{ur}(\sigma,\tau):=\frac{e^{-\beta H(\tau,\sigma)}}{\overleftarrow{Z}_{\sigma}}=\prod_{x\in\Lambda}\frac{e^{\beta h^{ur}_x(\sigma)\tau_x}}{2\cosh{\beta h^{ur}_x(\sigma) }}
$$

Note that in the definition of $H(\sigma,\tau)= - \sum_{x \in \Lambda}\tau_x h^{dl}(\sigma) - \lambda  \sum_{x \in \Lambda}\sigma_x$ the last term could be canceled obtaining the same value for the transition
probability $P^{dl}(\sigma,\tau)$. However we added it in the pair Hamiltonian for symmetry reasons: in particular the fact that $H(\tau,\sigma)$ is the
correct pair Hamiltonian to define $P^{ur}(\sigma,\tau)$ is due to this symmetry.
Note also that
 $$H(\sigma,\sigma) = H(\sigma)-q|\Lambda|$$
 where we define $H(\sigma)$ to be the usual Ising Hamiltonian with magnetic field $2\lambda$
\begin{equation}\label{HGibbs}
H(\sigma)=-\sum_{\{ x,y\}\in{\cal B}_\Lambda}J\sigma_x\sigma_y -2\lambda  \sum_{x \in \Lambda}\sigma_x
\end{equation}

We define 
\begin{equation}
P^{sh}(\sigma,\tau)=\sum_{\sigma'\in\cX_{\Lambda,B}}P^{dl}(\sigma,\sigma')P^{ur}(\sigma',\tau)=\sum_{\sigma'\in\cX_{\Lambda,B}}
\frac{e^{-\beta H(\sigma,\sigma')}}{\overrightarrow{Z}_\sigma}\frac{e^{-\beta H(\tau,\sigma')}}{\overleftarrow{Z}_{\sigma'}}
\end{equation}

Reversing the order of the ``down--left'' and the ``up--right'' updating rules one would obtain the chain with transition probabilities
$$
P^{sh'}(\sigma,\tau)=\sum_{\sigma'\in\cX_{\Lambda,B}}P^{ur}(\sigma,\sigma')P^{dl}(\sigma',\tau).
$$
Clearly, by choosing a different orientation  instead of down-left and up-right in $\ZZ^{2}$, a different pair Hamiltonian can be obtained
with a resulting different graph for the interaction. 

In this square case we could have directly used the alternate dynamics, since $\ZZ^2$ is
already a bipartite graph.
Indeed we can consider the checkerboard splitting of the sites in $\Lambda=V^{(1)}\cup V^{(2)}$, in black and white sites, with
$|V^{(1)}|=|V^{(2)}|=|V|=|\Lambda|/2$. Black
sites interact only with white sites and viceversa with the usual Ising Hamiltonian
\begin{align*}
H(\hbox{{\boldmath $ \sigma$}}) 
	& \equiv H(\sigma^{(1)},\sigma^{(2)})\\
	& = -\sum_{x\in V^{(1)}} \Big(\sigma^{(1)}_x h_x^{2\to 1}(\sigma^{(2)})+\lambda_x\sigma^{(2)}_x\Big)\\
	& = -\sum_{x\in V^{(2)}} \Big(\sigma^{(2)}_x h_x^{1\to 2}(\sigma^{(1)})+\lambda_x\sigma^{(1)}_x\Big).
\end{align*}

By Theorem \ref{t0} we immediately obtain that the invariant measure of
the alternate dynamics is the Gibbs measure $\pi^G(\hbox{{\boldmath $ \sigma$}})=e^{-\beta H(\hbox{{\boldmath $ \sigma$}})}/Z$.  The idea of alternate dynamics on even and odd sites is already present in the literature (see \cite{cirillo2002note}).

\subsection{Relation with the Gibbs measure}
\label{rel_G}

Remaining in $\Lambda\in\ZZ^2$ with $J>0$ and $B=\emptyset$, 
i.e. with the standard periodic boundary conditions, 
we denote the invariant measure of the shaken dynamics
$\pi_{\Lambda} =\frac{\overrightarrow{Z}_\sigma}{Z}$.

Let $\pi_{\Lambda}^G$ be the Gibbs measure
\begin{equation}\label{piG}
\pi_{\Lambda}^G(\sigma)=\frac{e^{-\beta H(\sigma)}}{Z^G}\hskip 1cm
{\rm{with}} \hskip 1cm Z^G=\sum_{\sigma\in\cX_{\Lambda}}e^{-\beta H(\sigma)}
\end{equation}
with $H(\sigma)$ defined in (\ref{HGibbs}) and we define  the total variation distance, or $L_1$ distance,
between two arbitrary probability measures $\mu$ and $\nu$ on $\cX_{\Lambda, B}$ as
\begin{equation}\label{dist}
\| \mu-\nu\|_{TV}=\frac{1}{2}\sum_{\sigma\in \cX_{\Lambda, B}}|\mu(\sigma)-\nu(\sigma)|
\end{equation}
In the following Theorem~\ref{t2} we control the distance between the invariant measure of the shaken
dynamics and the Gibbs measure
at low temperature (large $\beta$).
We notice that this theorem is an  extension of Theorem~1.2 in \cite{pss1} to the case of Hamiltonian with the non zero external field. 
This result could be extended to the case $B \neq \emptyset$.
\begin{theorem}\label{t2}
Set $\delta=e^{-2\beta q}$, and let $\delta$ be such that
\begin{equation}\label{conddelta}
\lim_{|\Lambda|\to\infty}\delta^2|\Lambda|=0
\end{equation}
Under the assumption (\ref{conddelta}), there exists $\bar{J}$  such that  for any $J>\bar{J}$
\begin{equation}\label{th2}
\lim_{|\Lambda|\to\infty}\| \pi_{\Lambda}-\pi_{\Lambda}^G\|_{TV}=0.
\end{equation}
\end{theorem}

\subsection{Convergence to equilibrium: a comparison}
\label{comp}
For the Ising model on the finite $L \times L$ box $\Lambda\subset \ZZ^2$ with periodic boundary conditions, consider now the regime 
$0 < \lambda < J$ and
very low temperature, i.e. $\beta\gg 1$. This is usually called ``metastable regime''.  Indeed the 
 configuration $-\mathbb{1}$
with all spins $-1$ represents, in this regime of low temperature, a metastable state, while the configuration $+\mathbb{1}$ with all positive spins, parallel to the external field $\lambda$, represents the stable state, where the Gibbs measure concentrates in the
limit $\beta\to\infty$.
We will call tunneling
    time $\tau_{+\mathbb{1}}$ 
the first hitting time to $+\mathbb{1}$ starting from $-\mathbb{1}$.

We compare here the tunneling times for different dynamics: the single spin flip (SSF) dynamics,
the PCA dynamics and the shaken dynamics (Sh) with $q < \lambda$.

More precisely consider,
for each $x \in \Lambda$, the local fields
\begin{align*}
    h^{SSF}_x(\sigma) & = \Big[\frac{J}{2}(\sigma_{x^{\uparrow}} + \sigma_{x^{\rightarrow}}+\sigma_{x^{\downarrow}} + \sigma_{x^{\leftarrow}})  + 2\lambda\Big] \\
    h_x^{PCA}(\sigma) & = \Big[\frac{J}{2}(\sigma_{x^{\uparrow}} + \sigma_{x^{\rightarrow}}+\sigma_{x^{\downarrow}} + \sigma_{x^{\leftarrow}})  +q\sigma_x + \lambda\Big]
\end{align*}
$h^{dl}_x(\sigma)$ and $h^{ur}_x(\sigma)$ defined in (\ref{hdl}) and (\ref{hur}),
and the local transition probabilities
$$
p^*_x(\sigma,\tau):=\frac{e^{\beta h^*_x(\sigma)\tau_x}}{2\cosh \beta h^*_x(\sigma)},\qquad *=SSF, PCA, dl,ur.
$$
Denote by $\sigma^x$  the configuration 
obtained from $\sigma$ by flipping the
spin at site $x$ and by $\indicator{A}$ the 
indicator function of event $A$. Then, the three dynamics under consideration, have
the following transition probabilities:

\begin{eqnarray*}
     P^{SSF}(\sigma,\tau) & = & \frac{1}{|\Lambda|}p^{SSF}_x(\sigma,\tau) \indicator{\tau = \sigma^x} \\
     P^{PCA}(\sigma,\tau) & = &\prod_{x\in \Lambda}p^{PCA}_x(\sigma,\tau)\\
     P^{Sh}(\sigma,\tau) & = & \sum_{\sigma'\in\cX_{\Lambda}}\prod_{x\in\Lambda}p^{dl}_x(\sigma,\sigma')
    \prod_{x\in\Lambda}p^{ur}_x(\sigma',\tau).
\end{eqnarray*}
\bigskip
The SSF dynamics is reversible with Gibbs invariant measure
$$\pi^G_{\Lambda}(\sigma)=\frac{e^{\beta\sum_x h^{SSF}_x(\sigma)\sigma_x}}{Z^G}$$
and the invariant measure of
PCA and shaken dynamics are, respectively,
\begin{equation*}
    \pi_{PCA}(\sigma) = \frac{\sum_{\tau} e^{\beta \sum_{x} h^{PCA}_{x}(\sigma) \tau_x}}{Z^{PCA}}
    \quad \text{ and }\quad
    \pi(\sigma) = \frac{\sum_{\tau} e^{\beta \sum_{x} h^{dl}_{x}(\sigma) \tau_x}}{Z}.
\end{equation*}

Note that, in the case $\frac{J}{2}=q$, the dynamics defined via the transition probabilities $P^{PCA}$ turns out to be the same as the \emph{cross} PCA dynamics introduced in \cite{cirillo2008metastability} and it is possible to show that the measure $\pi_{PCA}$ and the Gibbs-like measure with Hamiltonian $H(\sigma)$ defined in \cite[Equation~(2.5)]{cirillo2008metastability} coincide. 

All these measures,
in the regime of large $\beta$, concentrate on the stable state $+\mathbb{1}$. This is proven for $\pi(\sigma)_{PCA}$ in \cite{cirillo2008metastability} and it is immediate
for the Gibbs measure. 
For the shaken dynamics this follows by noting that,
for all $\sigma \neq +\mathbb{1}$, 
\begin{equation}
    \frac{\pi(\sigma)}{\pi(+\mathbb{1})} 
    = \prod_{x \in \Lambda} \frac{
        2 \cosh{\beta h_x^{dl}(\sigma)}
    } {
        2 \cosh{\beta [2J + q + \lambda]}
    } \le e^{-2 \beta q}
\end{equation}
and, hence
\[
    \lim_{\beta \to \infty} \frac{\pi(\sigma)}{\pi(+\mathbb{1})} = 0.
\]

For  large inverse temperature $\beta$ we have $p^*_x(\sigma,\tau)\sim 1$ if $\tau_x$ is parallel to the local field $h^*_x(\sigma)$. We call
such a local move ``along the drift''. 
On the other hand $p^*_x(\sigma,\tau)\sim e^{-2\beta |h^*_x(\sigma)|}$ if $\tau_x$ is anti-parallel to the local field $h^*_x(\sigma)$. We call
such a local move ``against the drift''. 

Let $\lambda < J$. For the \emph{SSF} dynamics we have for any $\delta>0$ (see for instance  \cite{OV,bovier2016metastability}):
\begin{equation}\label{t_SSF}
\lim_{\beta\to\infty}P^{SSF}_{-\mathbb{1}}(\tau_{+\mathbb{1}}> e^{\beta (\Gamma-\delta)})=1
\end{equation}
with 
$$
\Gamma=4J\ell_c-2\lambda \ell_c^2+2\lambda (\ell_c-1)
$$
and critical size $\ell_c=\big[\frac{J}{\lambda}\big]+1$, where $\big[\cdot\big]$
denotes the integer part.
The typical exit paths from the metastable state $-\mathbb 1$ follow a sequence of growing
squares and rectangles (quasi squares) of plus spins up to the critical size $\ell_c$. Starting from a rectangular
droplet of plus spins a move against the drift is necessary to create a new line, and the line  is 
completed with subsequent moves along the drift. 

A similar result for the \emph{PCA} dynamics is proven in \cite{cirillo2008metastability} in the case $q=\frac{J}{2}$.
In particular, the typical exit paths follow the same growth mechanism of the SSF dynamics since moves along the drift lead to rectangular droplets of plus spins and parallel updates against the drift
do not take place
with high probability for $\beta$ large (see \cite[Section~2.7]{cirillo2008metastability} for a detailed discussion).
\bigskip

A different growth takes place in the case of \emph{shaken dynamics}. Indeed using an argument inspired by \cite{dss2} it is simple to prove that configurations with complete
diagonals of plus spins can be used to construct a efficient way to go from the 
metastable to the stable state.
 
 We have the following
 
\begin{theorem}\label{t21}
For any $\delta>0$ and $0 < q < \lambda < J$ 
\begin{equation*}
\lim_{\beta\to\infty}P^{Sh}_{-\mathbb{1}}(\tau_{+\mathbb{1}}< e^{2\beta [(2J + q-\lambda)+\delta]})=1.
\end{equation*}
\end{theorem}

This means that the crossover takes place for the shaken dynamics, typically,
within a time corresponding to the 
time it takes, for the SSF and PCA,
to flip the first spin to $+1$. 

Moreover we can define an event characterizing the tube of typical 
paths from the metastable state $-\mathbb 1$
to the stable state $+\mathbb{1}$.
Actually we prove Theorem \ref{t21} by using this event.

We need some notation. Denote by $X^{Sh}_t$ the evolution of the shaken dynamics and let
$\tau_{-\mathbb{1}^c}=\inf\{t>0:X^{Sh}_t\not= {-\mathbb{1}}\} $ be the time
of the first spin flip, starting from $-\mathbb{1}$. We denote a site $x\in\Lambda$ as $x=(i,j)$ with $i,j=1,...,L$ and for each $k=1,...,L$
we define a diagonal $D_k\subset\Lambda$ as the set of sites forming a diagonal in the up-left to down-right direction on the torus 
(i.e., with periodic boundary condition):
\[
    D_k:=\{x=(i,j)\in \Lambda \hbox{ such that } (i+j-1) \equiv k \, (\mathrm{mod}\;L) \}
\]
 
Let
${\cal{D}}$ be the set of configurations 
where each up-left to down-right diagonal has
all spins with the same sign:
\[
    {\cal{D}}:=\{\sigma\in{\cal{X}}:\;\forall \, k=1,...,L \quad \sigma_x=\sigma_y\;  \hbox{ for any } x, y\in D_k\}.
\]
 Note that $-\mathbb{1}, +\mathbb{1}\in \cal{D}$.
 For notation convenience we identify a configuration  $\sigma\in{\cal{D}}$ with the set of indices of diagonals with positive spins,
 i.e., with a subset of $\{1,2,...,L\}$ denoted by $I_\sigma$. We say that $m\not\in I_\sigma$ is {\it nearest neighbor of} $I_\sigma$
 if there exists $n\in I_\sigma$ with $|n-m|=1$.
 Given $\sigma,\sigma'\in\cal{D}$ we say that $\sigma'$ {\it increases} $\sigma$ if  $I_{\sigma'}\supset I_\sigma$ and their difference $I_{\sigma'}\backslash I_\sigma$ is a single integer $m$ nearest neighbor of $I_\sigma.$
 
 Starting from $-\mathbb{1}$ define the sequence of strong times $\{s_i, t_i\}_{i=1,2,\ldots}$ as follows: $t_0=0$ and for any $i=1,2,..$
 \[
    s_i=\inf\{t>t_{i-1};\; X^{Sh}_t\not\in{\cal{D}}\},\qquad t_i=\inf\{t>s_i;\; X^{Sh}_t\in{\cal{D}}\}.
 \]
Fix an arbitrarily small positive $\delta$ and define $T_\delta=e^{2\beta [(2J + q-\lambda)+\delta]}$ and the following event
\begin{align*}
    {\cal{T}}:= 
        \big\{ & s_1=\tau_{-\mathbb{1}^c} < T_{\delta/2},\; 
                 t_L=\tau_{+\mathbb{1}} < T_\delta,\; 
                 \{ X^{Sh}_{t_i}\}_{i=1...L} \\
               & \hbox{ is a sequence of increasing configurations in }  
                 {\cal{D}}
            \big\}
\end{align*}

We have
\begin{theorem}\label{t22}
For any $\delta>0$
\[
    \lim_{\beta\to\infty}P^{Sh}_{-\mathbb{1}}({\cal{T}})=1.
\]
\end{theorem}

Clearly the event ${\cal{T}}$ implies the event $\tau_{+\mathbb{1}}<T_\delta$, so Theorem~\ref{t21} is established by proving Theorem~\ref{t22}.

The asymmetric nature of the interaction
gives the shaken dynamics
a higher mobility with respect to its
symmetric counterpart (``standard'' PCA).
This is the reason of shorter tunneling times.
Evidence of this fact,
in the regime $0 < q < \lambda < J$,
is provided
by the numerical experiment whose results are summarized in Figure~{\ref{fig:metastable_nucleation}}.
For both the PCA and the shaken dynamics
the simulations are started in configuration
$-\mathbb{1}$ and the value of the average
magnetization is tracked.
Simulations are run for several values of
the inverse temperature $\beta$.
It is immediately clear that, for the shaken dynamics,
the tunneling towards a state is much faster.
Note also that this higher mobility causes a slightly smaller magnetization at equilibrium
for finite $\beta$.

A similar behavior has been highlighted in \cite{LS, LSchap}
where a comparison between the symmetric
PCA and an irreversible PCA with totally asymmetric interaction has been performed in the case $\lambda = 0$.
In the case of null magnetic field,
if the value of the inverse temperature
$\beta$ is \emph{large} but finite, the
system exhibits, both in the case of PCA dynamics and shaken dynamics, a (positive or negative) spontaneous magnetization. 
On short time scales the system stays in one
``phase'' (e.g. it is negatively magnetized)
whereas, on longer time scales, it \emph{tunnels} rapidly towards the other phase.
However, since the value of the spontaneous magnetization is neither exactly $-1$ or $+1$,
choosing the right parameters to compare the two dynamics, as far as tunneling time is concerned, may be a somewhat delicate matter.
To make a fair comparison of the tunneling times
we tuned the parameters of the
PCA and the shaken dynamics so to have similar values for the (average) spontaneous magnetization. The results of the experiment
are summarized in Figure \ref{fig:comparison_mag_pca_shaken} where
the average magnetization over time is represented for both dynamics. It is
evident that, also in this case, 
the shaken dynamics 
has considerably more mobility than the PCA.
As far as the tunneling behavior is concerned, the shaken dynamics retains, essentially, the same features of the irreversible PCA considered in \cite{LSchap}. However, it is to remark that, in the case of the shaken dynamics, reversibility allows to manage the invariant measure more easily.

\begin{figure}
    \centering
    \subfigure[$\beta=1.3, q = 0.1, \lambda = 0.15$]
    {\includegraphics[width=0.48\textwidth]{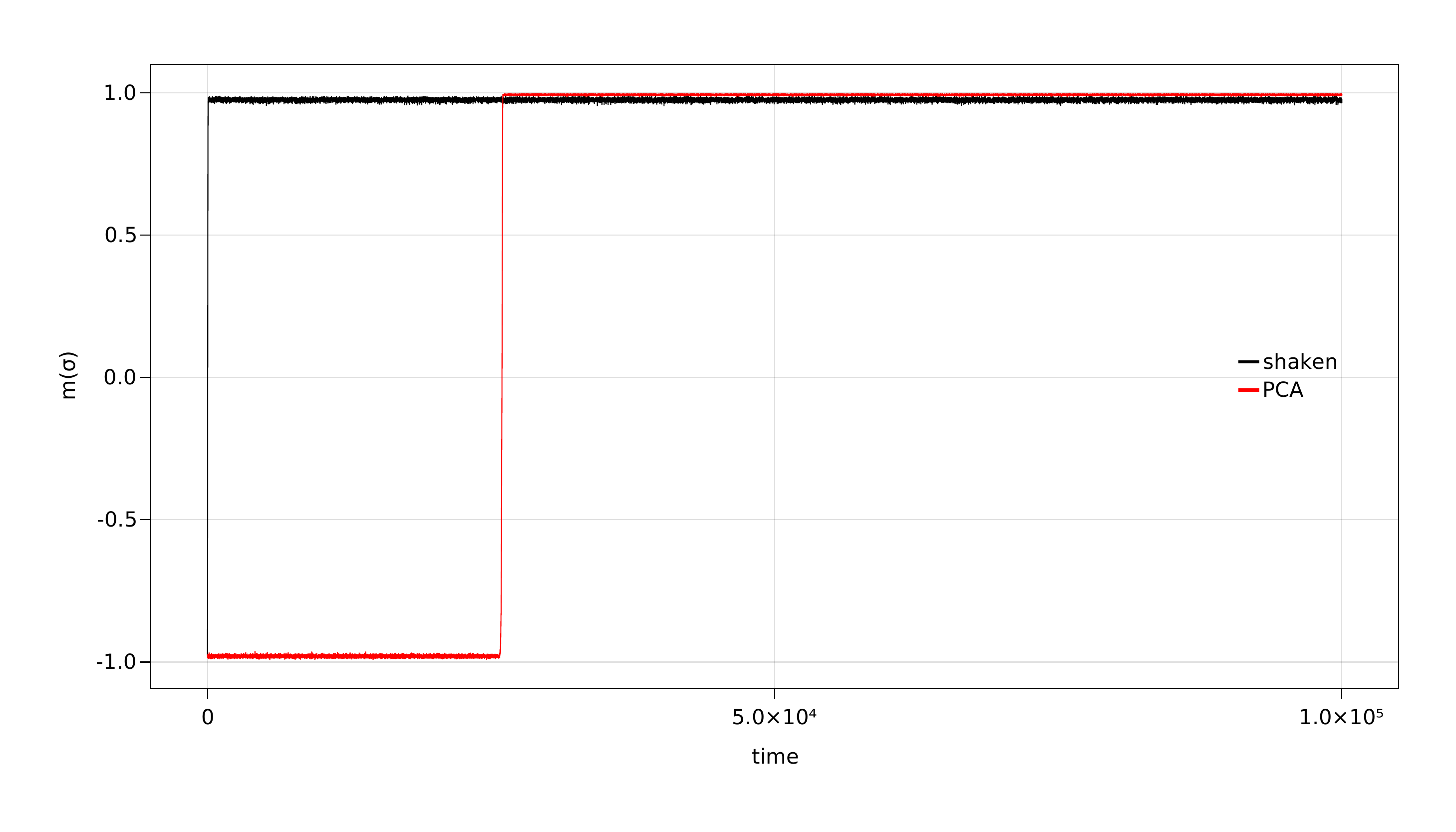}} \;
    \subfigure[$\beta=2.3, q = 0.1, \lambda = 0.15$]
    {\includegraphics[width=0.48\textwidth]{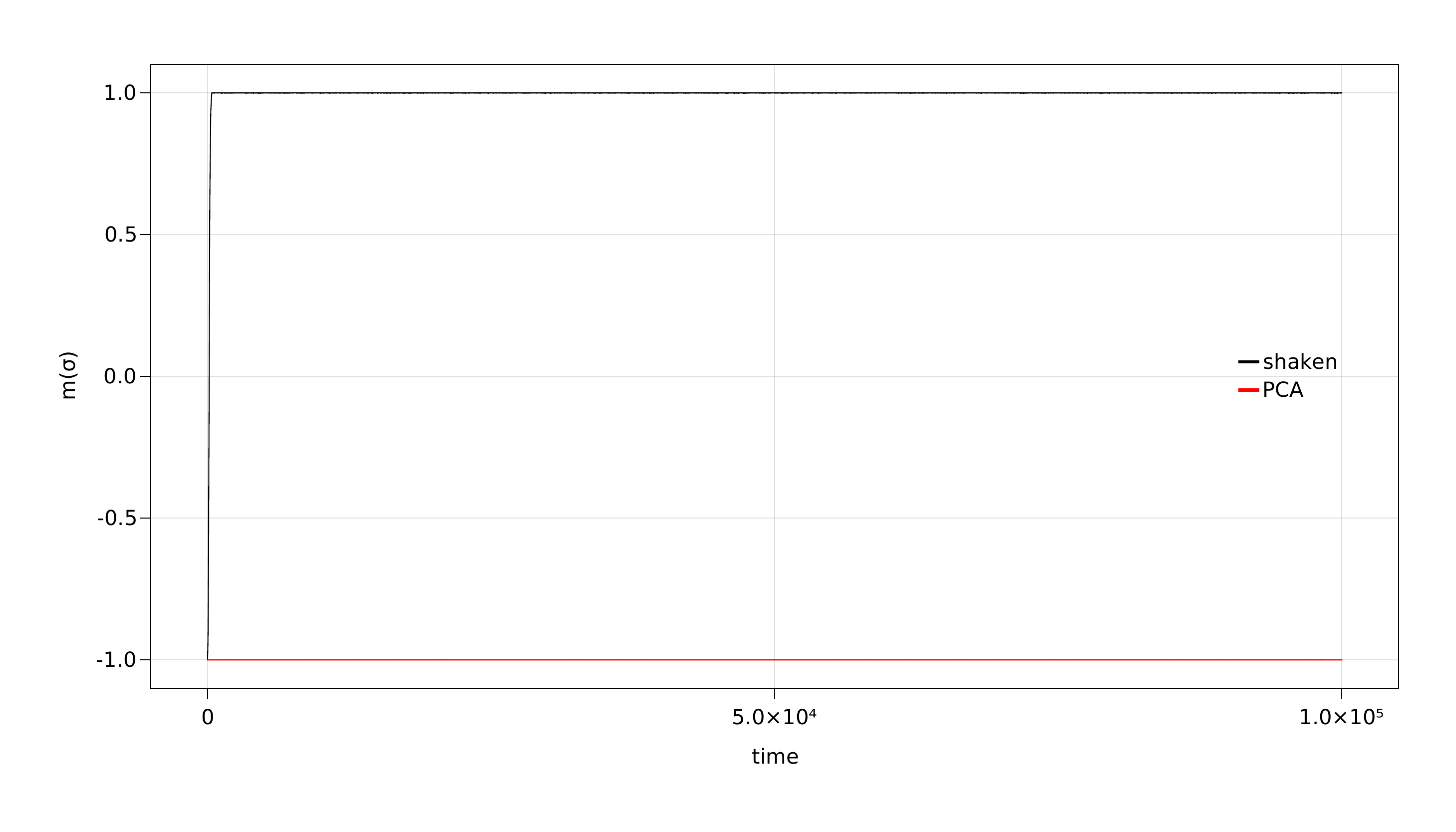}} \\
    \subfigure[$\beta=3.3, q = 0.1, \lambda = 0.15$]
    {\includegraphics[width=0.48\textwidth]{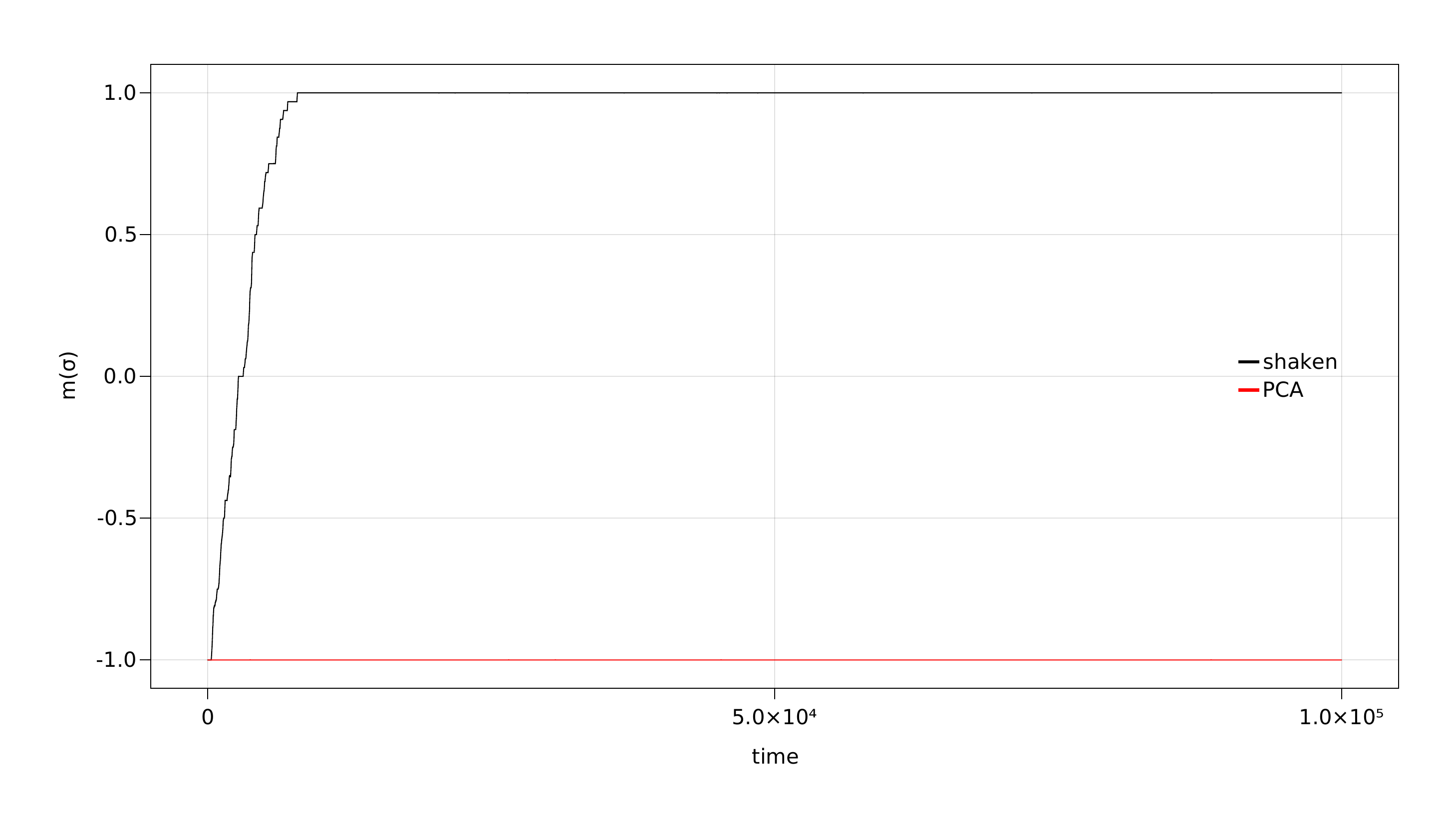}} \;
    \subfigure[$\beta=4.0, q = 0.1, \lambda = 0.15$]
    {\includegraphics[width=0.48\textwidth]{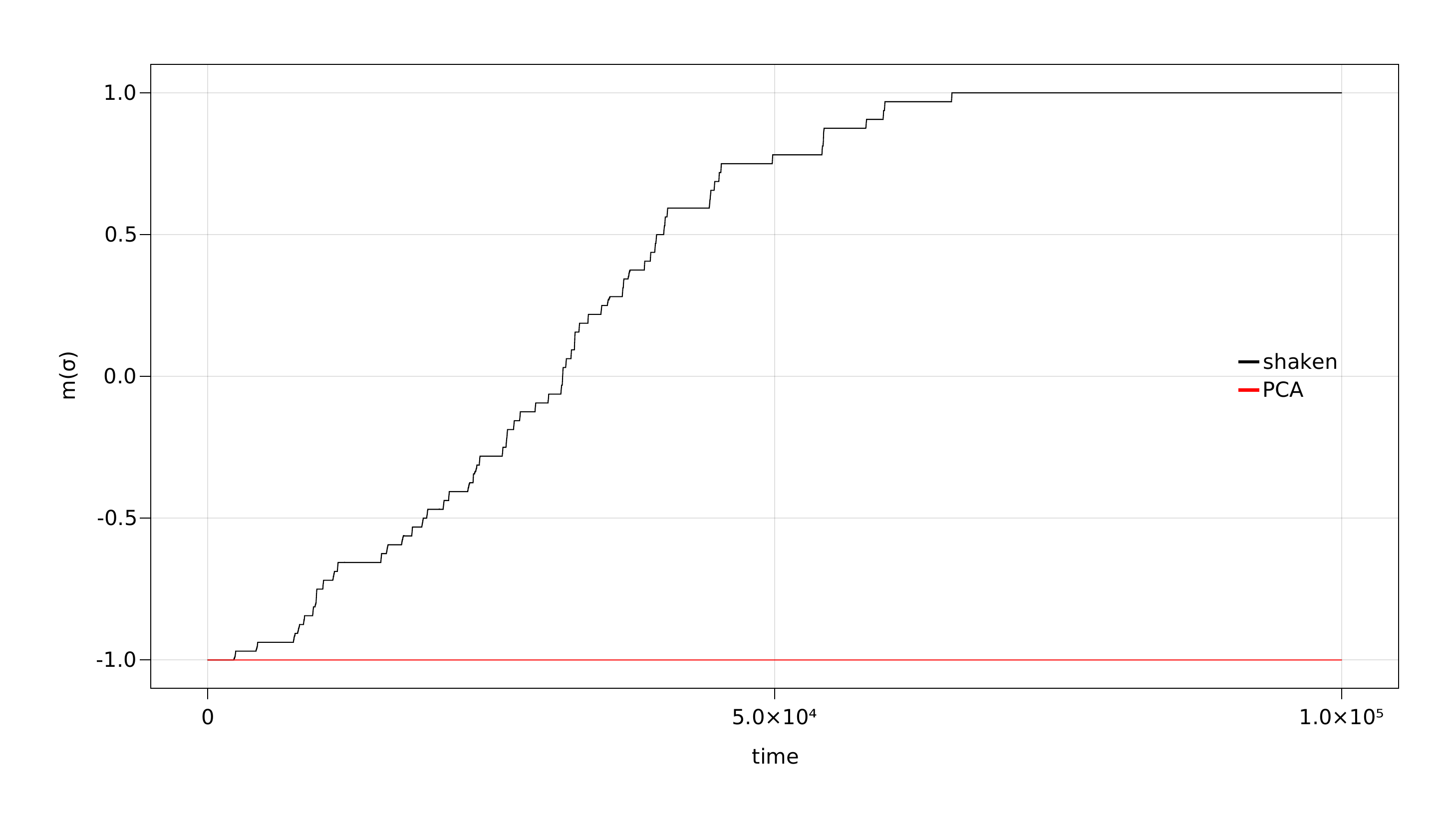}}
\caption{Comparison of the magnetization over time for PCA and shaken
dynamics for growing values of the inverse tempeature $\beta$.}
\label{fig:metastable_nucleation}
\end{figure}

\begin{figure}
    \centering
    \includegraphics[width=\textwidth]{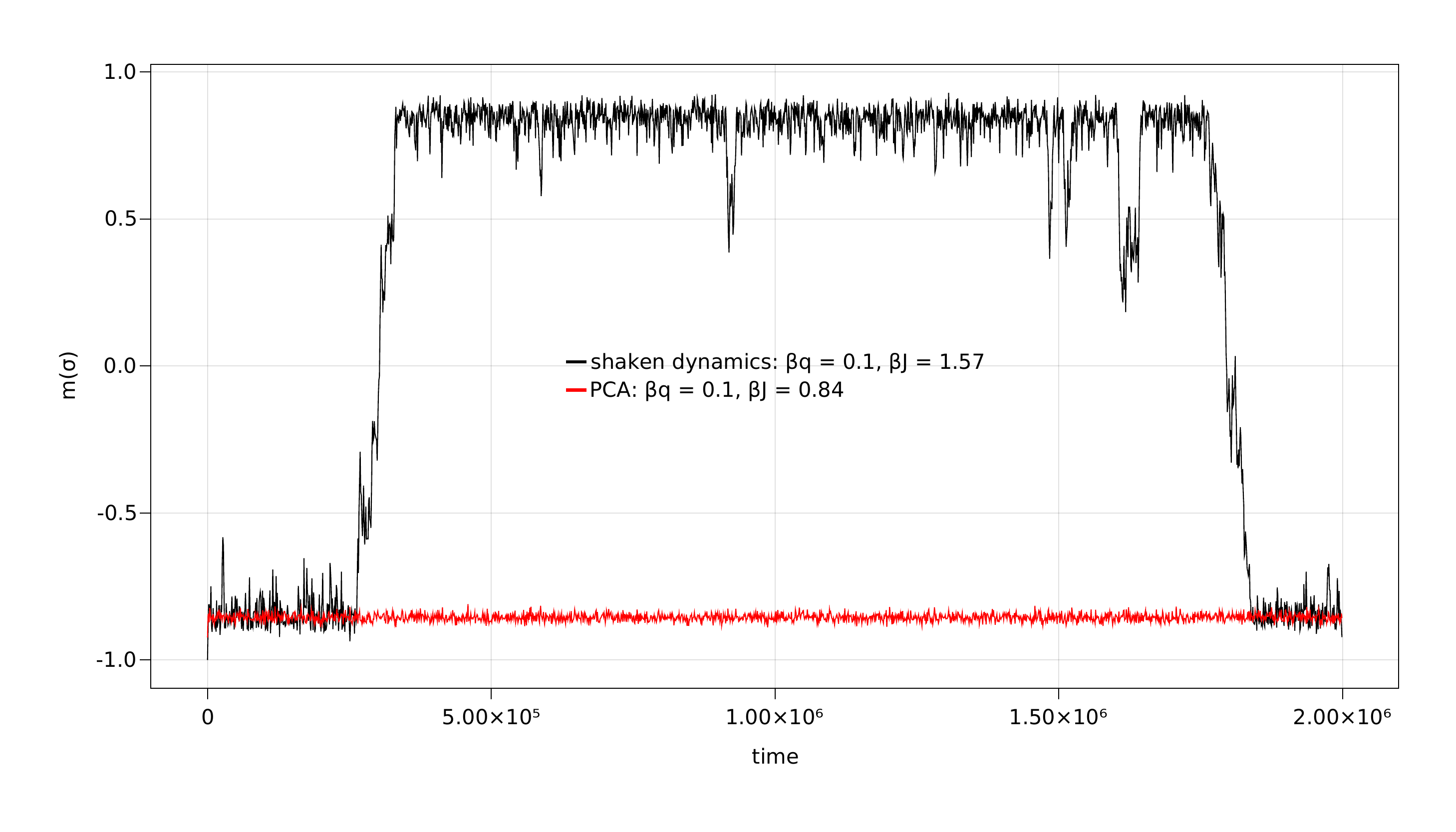}
    \caption{Comparison of the evolution of
    the magnetization for a spin system evolving
    according to a shaken dynamics (black) and according
    to a symmetric PCA (red). The values of the parameters
    are such that both dynamics exhibit the same spontaneous magnetization and are consistent
    with those of \cite{LSchap}.}
    \label{fig:comparison_mag_pca_shaken}
\end{figure}

\section{Geometrical discussion}
\label{geo_disc}

In the shaken dynamics the idea of alternate dynamics is combined with that of the doubling Hamiltonian.
Indeed considering only part of the interaction (for instance down-left first and then up-right in the case of $\Lambda\in\ZZ^2$  
presented at the beginning of the Section) and introducing the inertial parameter of self interaction $q$ it is possible to interpolate 
between different lattice geometries induced by the 
doubling graph as already discussed
in \cite{apollonio2019criticality}. 

Indeed the alternate dynamics on the hexagonal lattice 
makes possible to interpolate between the square ($q\to\infty$) and the 1-dimensional lattice ($q\to 0$).
The interpolation between lattices induced by the shaken dynamics may be applied in general, and in the case of planar
graphs the results concerning the critical behavior contained in \cite{apollonio2019criticality} can be extended, using  \cite{cimasoni2013critical}. 

Consider for instance the Ising model on the triangular lattice. On this lattice we divide the 6 nearest neighbors of each vertex $x$ into two sets, e.g. $\ell (x)$ left and $r(x)$ right nearest neighbors of $x$, and define a shaken dynamics with self interaction $q$. Hence the doubled Hamiltonian is
$$
H^{\triangle}(\sigma,\tau)=-\sum_x\Big[\sum_{y\in\ell (x)}\big(J\sigma_y\tau_x\big)+q\sigma_x\tau_x\Big]=-\sum_x\Big[\sum_{y\in r (x)}\big(J\tau_y\sigma_x\big)+q\sigma_x\tau_x\Big]
$$
The corresponding alternate dynamics turns out to be defined on the square lattice (see Fig.~\ref{fig:triangular_shaken})
with invariant measure the Gibbs one. In particular the square lattice is regular when we set $J=q$. In this case the parameter $q$ can be used to move through different geometries. The triangular lattice ($q\to\infty$) and the hexagonal lattice ($q = 0$) can be derived from the original square lattice just tuning the value of $q$.
A more precise statement of this interpolation is given by the following
\begin{theorem}\label{th:critical_curve_triangular_lattice}
    For the shaken dynamics on the triangular lattice the critical equation relating the parameters $J$ and $q$ is given by
    \begin{align}
        1 + \tanh^3(J)\tanh(q) = 3\tanh(J)\tanh(q) + 3\tanh^2(J)
    \end{align}
\end{theorem}

In the case $q=J$ we obtain the Onsager critical temperature for the square lattice, for $q = 0$ we obtain the critical temperature for the hexagonal lattice and in the limit $q \to \infty$ we obtain the critical temperature for the triangular lattice.

\begin{figure}
    \centering{
        \includegraphics[width=0.45\linewidth]{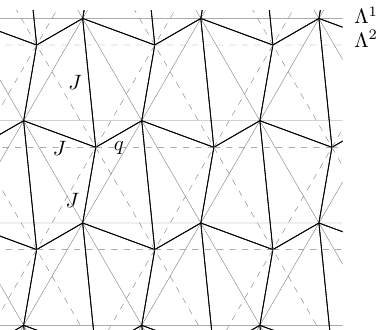}
    }
    \caption{Interaction in the pair Hamiltonian 
        for the shaken dynamics on the triangular lattice.
        Each spin of configuration $\sigma$ (living on the solid lattice) interacts with the spin at the same location and the three spins on its left in $\tau$ (living on the dashed lattice). The darker lines show that the pair interaction lives on a \emph{rectangular} lattice. For $q = J$ this lattice becomes the square one. As $q \to \infty$ the square lattice collapses onto the triangular lattice. If $q = 0$ the interaction graph becomes the
        homogeneous hexagonal lattice.
    }
    \label{fig:triangular_shaken}
\end{figure}

\section{The generalized shaken dynamics}
\label{gensh}

We can generalize the construction of the shaken dynamics.
Starting from a symmetric interaction ${\cal J}$ defining the Hamiltonian $H(\sigma)$, as in (\ref{ham1}),
we can define an arbitrary decomposition of the interaction matrix ${\cal J}$ in a sum of two matrices with
non negative entries
\begin{equation}\label{eq:Jd}
{\cal J}={\cal J}^o+{{\cal J}^o}^T.
\end{equation}
This means that every non oriented edge $\{x,y\}$ with weight $J_{xy}$ is decomposed in a pair
of oriented edges $(x,y)$ and $(y,x)$ with weight respectively ${\cal J}^o_{xy}$ and ${\cal J}^o_{yx}$.
Call $E^o$  the set of all these oriented edges
and apply the construction presented in Section \ref{Ggen} to construct the doubling graph by using this
set $E^o$ of oriented edges.

We proceed as  before defining the doubling Hamiltonian
\begin{align}\label{eq:generalized_hamiltonian}
\begin{aligned}
    H(\sigma^{(1)},\sigma^{(2)}) 
    & =
    -\scProd{\sigma^{(1)}, {{\cal J}^o}\sigma^{(2)}}
    -q\scProd{\sigma^{(1)},\sigma^{(2)}}
    -\scProd{\lambda,\sigma^{(1)}}
    -\scProd{\lambda,\sigma^{(2)}} \\
    & =
    -\scProd{{{\cal J}^o}^T\sigma^{(1)}, \sigma^{(2)}}
    -q\scProd{\sigma^{(1)},\sigma^{(2)}}
    -\scProd{\lambda,\sigma^{(1)}}
    -\scProd{\lambda,\sigma^{(2)}}.
\end{aligned}
\end{align}
In the case $\sigma^{(1)}=\sigma^{(2)}=\sigma$ by equation 
\eqref{eq:Jd} we have again
$H(\sigma,\sigma)=H(\sigma)-q|V|$.

The corresponding alternate dynamics
on the state space $\cX_V$  is defined with  two subsequent updating
as follows:
\begin{equation}
P^{1\to 2}(\sigma,\sigma'):=\frac{e^{-\beta H(\sigma,\sigma')}}{\overrightarrow{Z}_\sigma}, \qquad P^{2\to 1}(\sigma',\tau):=\frac{e^{-\beta H(\tau,\sigma')}}{\overleftarrow{Z}_{\sigma'}}
\end{equation}
and

\begin{equation}\label{gsh}
P^{sh}(\sigma,\tau)=\sum_{\sigma'\in\cX_V}P^{1\to 2}(\sigma,\sigma')P^{2\to 1}(\sigma',\tau)=\sum_{\sigma'\in\cX_V}
\frac{e^{-\beta H(\sigma,\sigma')}}{\overrightarrow{Z}_\sigma}\frac{e^{-\beta H(\tau,\sigma')}}{\overleftarrow{Z}_{\sigma'}}
\end{equation}

The results obtained in Theorem \ref{t0} can be immediately extended to this more general case.

The choice of the shaken dynamics discussed in Section \ref{Ggen}  is a particular case of generalized shaken dynamics
in which ${\cal J}^o_{xy}{\cal J}^o_{yx}=0$ for any pair $x,y$. 
In the general case the geometrical discussion of the doubling graph  of interaction is much more complicated. Also the interpolation between different geometries obtained for different values of the parameter $q$, as discussed  in Section \ref{Z2}, is  more involved in this generalized case.
\begin{figure}[tbh]
\centering
  \includegraphics[width=0.7\linewidth]{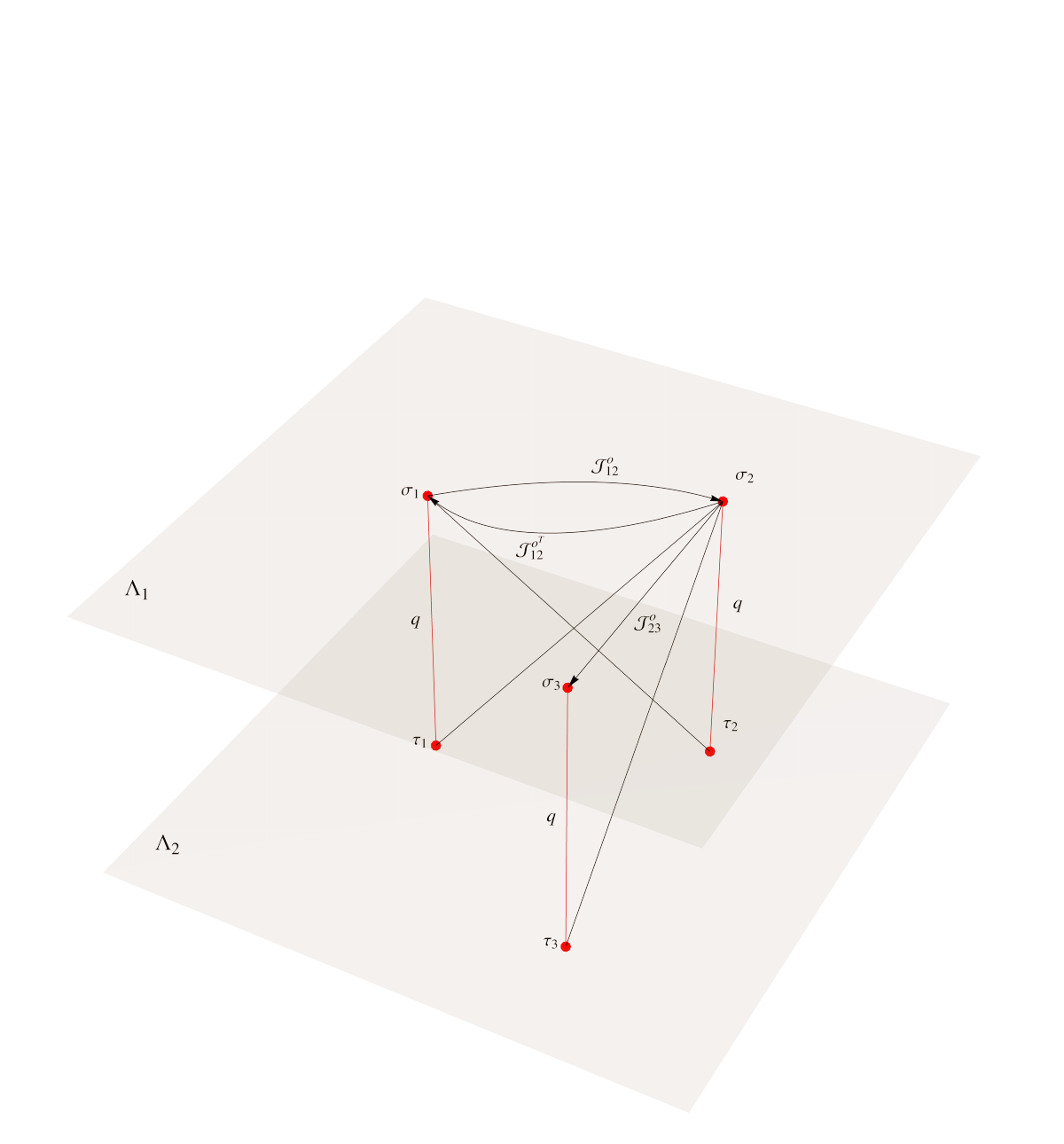}
    \caption{The construction of the doubling graph in the case of the generalized shaken dynamics.}
\end{figure}

Another particular choice in this class of generalized shaken dynamics is ${\cal J}^o=\frac{1}{2}{\cal J}$
corresponding to the PCA discussed in \cite{dss1}.

\section{Application to optimization problems}
\label{sec:application_to_optimization_problems}

The shaken dynamics on a general graph, and its generalization,
gets the possibility to look for the minimum of a general Hamiltonian $H(\sigma)$ defined on $\{-1, +1\}^V$ by means of a
parallel dynamics, by using the following result that could be considered a corollary of Theorem \ref{t0}.
In combinatorial optimization this can be used as a parallel approach to the Quadratic Unconstrained Binary Optimization (QUBO)
i.e., the problem of minimizing a quadratic polynomial of binary variables (see \cite{UBQP} for a survey). 
 
\begin{theorem}\label{c1}\label{corollary_q>}
Consider  a Hamiltonian $H(\sigma)$ of the form given in \eqref{eq:generalized_hamiltonian} on $\{-1, +1\}^V$, for any doubling Hamiltonian $H(\sigma,\tau)$ defined in \eqref{Hs},
corresponding to a bipartite graph $G^b=(V^b, E^b)$.
\begin{enumerate}[label=(\arabic*)]
    \item If 
    \begin{align}\label{q>}
    q>\max_{x\in V}\Big\{
    \sum_{y: \{x,y\}\in E^b}|J_{xy}|+|\lambda_x|
    \Big\}
    \end{align}
    then 
    the alternate dynamics defined with $H(\sigma,\tau)$ is a parallel algorithm to find configurations $\sigma$ minimizing $H(\sigma)$.
    \item For any positive $q$ the alternate dynamics
    defined with $H(\sigma, \tau)$ is a parallel
    algorithm to find a lower bound for
    $\min_{\sigma} H(\sigma)$.
\end{enumerate}

\end{theorem}

To assess the effectiveness of the strategy presented in Theorem~\ref{corollary_q>}, we put forward some preliminary tests on a 
simplified version of the Edwards-Anderson model where the weight of the edges connecting neighboring sites is set to $J=+1$ with
probability $\frac{1}{2}$ and $J=-1$ with probability $\frac{1}{2}$ and where the external field is zero. 
In this case, setting $q>2$ is sufficient to satisfy the hypotheses of the corollary. We compared the results with those obtained
with a single spin flip heat bath dynamics and considered ``grids'' with side length 128 and 256.
With this setting, the heuristic minima that we obtained with the shaken dynamics are essentially equivalent to those obtained
with the single spin flip dynamics. However the speed up with respect to the single spin flip dynamics was significant. To be as fair as possible in this comparison, we renormalized the time of the single spin flip dynamics with the number of vertices in the graph so to have the same number of ``attempted spin flips''. We observed a speed-up of
about 10 times when considering, for both algorithms, a CPU implementation and up to 200 times when comparing the CPU implementation
of the single spin flip dynamics with a GPU implementation of the shaken dynamics. 
These findings are in agreement with Theorem~\ref{t21}.

\section{Proofs of the results}
\label{sec:proofs}

\subsection{Proof of Theorem~\ref{t0}}
\label{p1}
We have immediately the  detailed balance condition w.r.t. the measure $\pi(\sigma)$
indeed
\begin{equation}\label{dbc}
\sum_{\sigma'\in\cX_V}\frac{e^{-\beta (H(\sigma,\sigma')+H(\tau,\sigma'))}}{\overleftarrow{Z}_{\sigma'}}=\overrightarrow{Z}_\sigma P^{sh}(\sigma,\tau)= \overrightarrow{Z}_\tau P^{sh}(\tau,\sigma)=
\sum_{\sigma'\in\cX_V}\frac{e^{-\beta (H(\tau,\sigma')+H(\sigma,\sigma'))}}{\overleftarrow{Z}_{\sigma'}}
\end{equation}

It is straightforward to prove that ${\pi^b}(\sigma^{(1)},\sigma^{(2)})$ is the stationary measure of $P^{alt}$
\begin{align}\label{stmeasalt}
\begin{aligned}
\sum_{\sigma^{(1)},\sigma^{(2)}}{{\pi^b}}(\sigma^{(1)},\sigma^{(2)})P^{alt}(\hbox{{\boldmath $ \sigma$},{\boldmath $ \tau$}}) 
	& = \sum_{\sigma^{(1)},\sigma^{(2)}}\frac{e^{-\beta H(\sigma^{(1)},\sigma^{(2)})}}{Z}
								\frac{e^{-\beta H(\sigma^{(1)},\tau^{(2)})}}{\overrightarrow{Z}_{\sigma^{(1)}}}
								\frac{e^{-\beta H(\tau^{(1)},\tau^{(2)})}}{\overleftarrow{Z}_{\tau^{(2)}}}\\
	&=\frac{e^{-\beta H(\tau^{(1)},\tau^{(2)})}}{Z}={\pi^b}(\tau^{(1)},\tau^{(2)})
\end{aligned}
\end{align}
$\hfill\qed$

Note that, in general
$$
{{\pi^b}}(\sigma^{(1)}, \sigma^{(2)})P^{alt}(\hbox{{\boldmath $ \sigma$},{\boldmath $ \tau$}})\not=
{{\pi^b}}(\tau^{(1)}, \tau^{(2)})P^{alt}(\hbox{{\boldmath $ \tau$},{\boldmath $ \sigma$}}).
$$
For instance consider the bipartite graphs $K_{n,n}$ with equal weights on all edges and where, for all $i$,
$(\sigma_{i}^{(1)}, \sigma_{i}^{(2)}) = (+1, +1)$ and $(\tau_{i}^{(1)}, \tau_{i}^{(2)}) = (+1, -1)$. 

\subsection{Proof of Theorem~\ref{t2}}
\label{p2}
To prove Theorem~\ref{t2} it is possible to argue as in the proof of Theorem~1.2 in \cite{pss1}.

In our notation $\pi_{\Lambda}$ and $\pi_{\Lambda}^G$ have the role,
respectively, of $\pi_{PCA}$ and $\pi_{G}$ used in \cite{pss1}.
Further let $g_{x}(\sigma) := J(\sigma_{x^{\downarrow}} + \sigma_{x^{\leftarrow}})$
be the analogue of $h_i(\sigma)$ in \cite{pss1}. Here we assume $\lambda < 0$. 
The case $\lambda > 0$ can be treated likewise.

Recalling that $\delta = e^{-2 \beta q}$,
it is possible to write $\overrightarrow{Z}_{\sigma}$ in the following way:
\begin{align}\label{eq:unpacked_pca_measure}
    \begin{aligned}
    \overrightarrow{Z}_{\sigma} & = \sum_{\tau}e^{-\beta H(\sigma, \tau)}
                       = \sum_{\tau}e^{-\beta H(\sigma, \sigma)}e^{-\beta [H(\sigma, \tau) - H(\sigma, \sigma)]}\\
                     & = e^{\beta q|\Lambda|}e^{-\beta H(\sigma)}
                        \sum_{\tau}e^{\beta \sum_{x : \sigma_x \neq \tau_x} -2g_{x}(\sigma)\sigma_{x}-2q - 2\lambda\sigma_{x}} \\
                     & = e^{\beta q|\Lambda|}e^{-\beta H(\sigma)}
                        \sum_{I\subset \Lambda}\delta^{|I|}\prod_{x \in I}e^{-\beta(2g_{x}(\sigma)\sigma_x + 2\lambda \sigma_{x})} \\
                      & = e^{\beta q|\Lambda|}e^{-\beta H(\sigma)}
                        \prod_{x \in \Lambda}(1 + \delta e^{-\beta (2g_{x}(\sigma)\sigma_x + 2\lambda \sigma_{x})})
    \end{aligned}
\end{align}
where the sum over $\tau$
has been rewritten as the sum over all subsets
$I \subset \Lambda$ such that $\tau_x = -\sigma_x$ if $x \in I$
and  $\tau_x = \sigma_x$ otherwise.
The factor $e^{\beta q|\Lambda|}$ does not depend on $\sigma$ and cancels out
in the ratio $\frac{\overrightarrow{Z}_{\sigma}}{Z}$.

Call $f(\sigma) := \prod_{x \in \Lambda}(1 + \delta e^{-\beta (2g_{x}(\sigma)\sigma_x + 2\lambda\sigma_{x})})$, \,
$w(\sigma) := e^{-\beta H(\sigma)}f(\sigma)= w^{G}(\sigma)f(\sigma)$. Then, by \eqref{eq:unpacked_pca_measure},
it follows 
\begin{align*}
    \pi_{\Lambda}(\sigma) & = \frac{w(\sigma)}{\sum_{\tau}w(\tau)} = \frac{w^G(\sigma)f(\sigma)}{\sum_{\tau}w^G(\tau)f(\tau)}
     = \frac{\frac{w^G(\sigma)}{Z^G}f(\sigma)}{\sum_{\tau}\frac{w^G(\tau)}{Z^G}f(\tau)}
      = \frac{\pi_{\Lambda}^G(\sigma)f(\sigma)}{\pi_{\Lambda}^G(f)}
\end{align*}
with $\pi_{\Lambda}^G(f) = \sum_{\sigma} \pi_{\Lambda}^G(\sigma)f(\sigma)$.

As in \cite{pss1}, using Jensen's inequality
the total variation distance between $\pi_{\Lambda}$ and $\pi_{\Lambda}^G$
can be bounded as
\begin{align*}\label{eq:tv_distance_bound}
	\| \pi_{\Lambda}-\pi_{\Lambda}^G\|_{TV}
						  \le \sqrt{\frac{\pi_{\Lambda}^G(f^2)}{(\pi_{\Lambda}^G(f))^2} - 1} =: \sqrt{(\Delta(\delta))}.
\end{align*}

To prove the theorem, it will be shown that
$\Delta(\delta) = O(\delta^2|\Lambda|)$.

By writing
$\Delta(\delta) = e^{\log(\pi_{\Lambda}^G(f^2)) - 2\log(\pi_{\Lambda}^G(f))} - 1,$
the claim follows by showing that the argument of the exponential divided by $|\Lambda|$ is analytic in $\delta$ and that the first order term of its expansion in $\delta$ cancels out.

In other words the claim follows thanks to the following lemma.
\begin{lemma}\label{lemma_analiticity_in_delta}
    There exists $J_{c}$ such that, for all $J > J_{c}$
    \begin{enumerate}
        \item $\frac{\log(\pi_{\Lambda}^G(f^2))}
                    {|\Lambda|}$ and
              $\frac{\log(\pi_{\Lambda}^G(f))}
                    {|\Lambda|}$ are analytic in $\delta$ for $|\delta| < \delta_{J}$
        \item $\frac{\log(\pi_{\Lambda}^G(f^2))}{|\Lambda|} - 2\frac{\log(\pi_{\Lambda}^G(f))}{|\Lambda|} = O(\delta^{2})$
    \end{enumerate}
\end{lemma}

\begin{proof}
    The analyticity of $\frac{\log(\pi_{\Lambda}^G(f^2))}{|\Lambda|}$ and
    $\frac{\log(\pi_{\Lambda}^G(f))}{|\Lambda|}$ is proven by showing that these quantities
    can be written as partition functions of an abstract polymer gas. The analyticity is obtained
    using standard cluster expansion.

    To carry over this task, we will rewrite $\pi_{\Lambda}^G(f^{k})$
    in terms of standard Peierls contours.
    Divide the sites in $\Lambda$ according to the value of the spins and number
    of edges of the Peierls contour left and below the site in the following way:
    \begin{itemize}[label=\adfbullet{43}]
        \item $\spinclassA$: $\{x \in \Lambda :
                            \sigma_{x} = -1 \land
                            (\sigma_{x^{\leftarrow}}= -1, \sigma_{x^{\downarrow}}= -1)\}$;
        \item $\spinclassB$: $\{x \in \Lambda :
                            \sigma_{x} = -1 \land
                            ((\sigma_{x^{\leftarrow}}= +1, \sigma_{x^{\downarrow}}= -1) \lor
                            (\sigma_{x^{\leftarrow}}= -1, \sigma_{x^{\downarrow}}= +1))\}$;
        \item $\spinclassC$: $\{x \in \Lambda :
                                \sigma_{x} = -1 \land
                                \sigma_{x^{\leftarrow}}= +1, \sigma_{x^{\downarrow}} = +1\}$;
        \item $\spinclassD$: $\{x \in \Lambda :
                            \sigma_{x} = +1 \land
                            (\sigma_{x^{\leftarrow}}= +1, \sigma_{x^{\downarrow}}= +1)\}$;
        \item $\spinclassE$: $\{x \in \Lambda :
                            \sigma_{x} = +1 \land
                            ((\sigma_{x^{\leftarrow}}= +1, \sigma_{x^{\downarrow}}= -1) \lor
                            (\sigma_{x^{\leftarrow}}= -1, \sigma_{x^{\downarrow}}= +1))\}$;
        \item $\spinclassF$: $\{x \in \Lambda :
                                \sigma_{x} = +1 \land
                                (\sigma_{x^{\leftarrow}}= -1, \sigma_{x^{\downarrow}}= -1)\}$;
    \end{itemize}
    With this notation, $f(\sigma)$ can be written as
    \begin{align}
    \begin{aligned}
        f(\sigma) & = [1+\delta e^{\beta(-4J + 2\lambda)}]^{|\Lambda|}
                   \prod_{x \in \spinclassB}\frac{(1+\delta e^{+ 2\beta \lambda})}
                                                  {[1+\delta e^{\beta(-4J + 2\lambda)}]}
                    \prod_{x \in \spinclassC}\frac{[1+\delta e^{\beta(4J + 2\lambda)}]}
                                                  {[1+\delta e^{\beta(-4J + 2\lambda)}]}\\
                  & \prod_{x \in \spinclassD}\frac{[1+\delta e^{\beta(-4J - 2\lambda)}]}
                                                               {[1+\delta e^{\beta(-4J + 2\lambda)}]}
                    \prod_{x \in \spinclassE}\frac{(1+\delta e^{-2 \beta \lambda})}
                                                               {[1+\delta e^{\beta(-4J + 2\lambda)}]}
                    \prod_{x \in \spinclassF}\frac{[1+\delta e^{\beta(4J - 2\lambda)}]}
                                                               {[1+\delta e^{\beta(-4J + 2\lambda)}]} \\
                  & = [1+\delta e^{\beta(-4J + 2\lambda)}]^{|\Lambda|} \tilde{\xi}(\sigma, \lambda)
    \end{aligned}
    \end{align}
    with
    \begin{align}
    \begin{aligned}
                \tilde{\xi}(\sigma, \lambda) = &
                    \left[\frac{(1+\delta e^{+2\beta \lambda})}
                               {[1+\delta e^{\beta(-4J + 2\lambda)}]} \right]^{\numspinsB}
                    \,\times\,
                    \left[\frac{[1+\delta e^{\beta(4J + 2\lambda)}]}
                               {[1+\delta e^{\beta(-4J + 2\lambda)}]}\right]^{\numspinsC} 
                    \,\times\,\\
                &  \left[\frac{[1+\delta e^{\beta(-4J - 2\lambda)}]}
                               {[1+\delta e^{\beta(-4J + 2\lambda)}]}\right]^{\numspinsD}
                    \,\times\,
                    \left[\frac{(1+\delta e^{-2\beta \lambda})}
                               {[1+\delta e^{\beta(-4J + 2\lambda)}]}\right]^{\numspinsE}
                    \,\times\, \\
                &    \left[\frac{[1+\delta e^{\beta(4J - 2\lambda)}]}
                               {[1+\delta e^{\beta(-4J + 2\lambda)}]}\right]^{\numspinsF}
    \end{aligned}
    \end{align}

    For a given a configuration $\sigma\in\cX_{\Lambda}$, we denote by $\gamma(\sigma)$ its Peierls contour
    in the dual ${\cal B}_\Lambda^*=\cup_{(x,y)\in {\cal B}_\Lambda}(x,y)^*$
    \begin{equation}\label{Pc1}
        \gamma(\sigma):= \{(x,y)^*\in {\cal B}_\Lambda^*:\; \sigma_x\sigma_y=-1\}
    \end{equation}

    Noting that $e^{-\beta H(\sigma)} = e^{(2J - 2\lambda)\beta |\Lambda|} e^{-2\beta J|\gamma(\sigma)| +4\beta  \lambda\numpluses{\sigma}}$, with
    $\numpluses{\sigma} = \sum_{x \in \Lambda} 1_{\{\sigma_{x} = +1\}}$
    is the number of plus spins in $\Lambda$ of configuration $\sigma$, we have
    \begin{align}
        \pi_{\Lambda}^G(f^{k}) = \frac{1}{Z^{G}}e^{(2J - 2\lambda)\beta |\Lambda|}
                            [1+\delta e^{\beta(-4J + 2\lambda)}]^{k|\Lambda|}
                            \sum_{\sigma}\left[e^{-2\beta J|\gamma(\sigma)|+ 4\beta \lambda\numpluses{\sigma}} \tilde{\xi}(\sigma, \lambda)^k
                            \right]
    \end{align}

    Setting
    \begin{align}
    \begin{aligned}
                \xi(\sigma, \lambda)  = &
                    \left[\frac{(1+\delta e^{+2\beta \lambda})}
                            {[1+\delta e^{\beta(-4J + 2\lambda)}]} \right]^{\numspinsB} \,\times\,
                    \left[\frac{[1+\delta e^{\beta(4J + 2\lambda)}]} \,\times\,
                            {[1+\delta e^{\beta(-4J + 2\lambda)}]}\right]^{\numspinsC} \,\times\,\\
                &  \left[\frac{e^{+2\beta \lambda}[1+\delta e^{\beta(-4J - 2\lambda)}]}
                            {[1+\delta e^{\beta(-4J + 2\lambda)}]}\right]^{\numspinsD} \,\times\,
                    \left[\frac{e^{+2\beta \lambda}(1+\delta e^{-2\beta \lambda})}
                            {[1+\delta e^{\beta(-4J + 2\lambda)}]}\right]^{\numspinsE} \,\times\, \\
                &    \left[\frac{e^{+2\beta \lambda}[1+\delta e^{\beta(4J - 2\lambda)}]}
                            {[1+\delta e^{\beta(-4J + 2\lambda)}]}\right]^{\numspinsF}
    \end{aligned}
    \end{align}
    allows us to write, for $k \in \{1, 2\}$,
    \begin{align}
        \sum_{\sigma}\left[e^{-2\beta J|\gamma(\sigma)|+ 4\beta \lambda\numpluses{\sigma}} \tilde{\xi}(\sigma, \lambda)^k
        \right] =
        \sum_{\sigma}\left[e^{-2\beta J|\gamma(\sigma)|}\left(e^{+{2\beta \lambda\numpluses{\sigma}}}\right)^{2-k} \xi(\sigma, \lambda)^k
        \right]
    \end{align}

    A straightforward computation yields
    $\xi(\sigma, \lambda)^k \le \xi(\sigma, 0)^k$ and then
    \begin{equation}
    \begin{aligned}
        \sum_{\sigma} \left[
            e^{-2\beta J|\gamma(\sigma)|} \left(
                    e^{+{2\beta \lambda\numpluses{\sigma}}}
            \right)^{2-k} \xi(\sigma, \lambda)^k 
        \right]
        &d\leq \sum_{\sigma} e^{-2\beta J|\gamma(\sigma)|} \xi(\sigma, 0)^k \\
        & = 2\sum_{\gamma}e^{-2\beta J|\gamma|} \xi(\gamma, 0)^k
    \end{aligned}
    \end{equation}
    
    where $\xi(\gamma, 0)^k$ coincides with $\xi^I_k(\Gamma)$ in the proof of Lemma 2.3 in \cite{pss1},
    with $\numspinsB + \numspinsE = |l_{1}(\Gamma)|$ and
    $\numspinsC + \numspinsF = |l_{2}(\Gamma)|$.

    This implies that the proof can be concluded following the same steps as in \cite{pss1}.
\end{proof}

\subsection{Proof of Theorem ~\ref{t22}}

Starting from $-\mathbb{1}$ the first spin flip is a move against the drift having a probability
$e^{-2\beta |h^{Sh}_x(-\mathbb{1})|}= e^{-2\beta(2J+q-\lambda)}$ for some $x\in\Lambda$. Different exits from $-\mathbb{1}$,
with more than a single spin flip, have a probability exponentially smaller than this.
Hence, denoting for simplicity a configuration with the set of its plus spins, we have $X^{Sh}_{\tau_{-\mathbb{1}^c}}=\{x\}$ and
therefore
\[
\lim_{\beta\to\infty}P^{Sh}_{-\mathbb{1}} 
    \Big(\{\tau_{-\mathbb{1}^c}<T_{\delta/2}\} \cap \{X^{Sh}_{\tau_{-\mathbb{1}^c}}=\{x\} \, \hbox{ for some } x\in\Lambda\}\Big) = 1.
\]
 Starting now from this single plus spin in $x=(i,j)$ we have a sequence of $[L/2]$ moves along the drift producing a configuration in
 ${\cal{D}}$ with plus spin in the diagonal containing $x$, i.e. $D_k$ with $k={(i+j-1) \mathrm{mod}\;L}$. Indeed in the first half ($dl$) step of the dynamics the configuration $\{(i,j+1),(i+1,j)\}$ is reached with a move along the drift and in the subsequent ($ur$) half step the configuration $\{(i,j),(i-1,j+1),(i+1, j-1)\}$ is reached 
 and so on up to reach in  $[L/2]$ steps along the drift the configuration $\omega_1=\omega_1(x)\in {\cal D}$ with plus spins in $D_k$. 
 We obtain
\[
\lim_{\beta\to\infty}P^{Sh}_{-\mathbb{1}} 
    \Big(\{\tau_{-\mathbb{1}^c}<T_{\delta/2}\} 
    \cap \{X^{Sh}_{\tau_{-\mathbb{1}^c}}=x \text{ \small{for some} } x\in\Lambda\}
    \cap \{ X^{Sh}_{t_1}=\omega_1(x)\}
    \cap \{t_1<2 T_{\delta/2}\}\Big)= 1 
\]

The minimal cost to leave a configuration in ${\cal{D}}$, i.e., the minimal $|h^{Sh}_.(\cdot)|$ to pay in the exponent of the probability,
is $2J-\lambda-q$ corresponding to the cost of a flip to plus in a site adjacent to the diagonal. Indeed to flip to minus
a spin in the diagonal has a cost $2J+\lambda-q$. To flip a spin not adjacent to the diagonal has a larger cost. This implies
that in a subsequent interval of time $e^{2\beta(2J-\lambda-q+\delta)}$ an increased configuration $\omega_2\in{\cal D}$ is
reached with $\omega_2=\omega_1^+$ or $\omega_2=\omega_1^-$ and  $I_{\omega_1^\pm}=k\pm1$.
Note that, by the same arguments, it follows that the state $-\mathbb{1}$ is indeed metastable, i.e., 
it has the largest stability level
in the sense of \cite{MNOS}.
 
We get 
\begin{align*}
\lim_{\beta\to\infty}P^{Sh}_{-\mathbb{1}} \Big(
    & \{\tau_{-\mathbb{1}^c}<T_{\delta/2}\} 
      \cap \{X^{Sh}_{\tau_{-\mathbb{1}^c}}=x \, \hbox{ for some } x\in\Lambda\}
      \cap \{ X^{Sh}_{t_1}=\omega_1(x)\} \\
    & \cap \{X^{Sh}_{t_2}=\omega_1(x)^+ \hbox{ or } \omega_1(x)^-\}\cap \{t_2<3T_{\delta/2}\}
    \Big) = 1.
\end{align*}

By iterating this argument $L-2$ times we get the result.
\hfill $\qed$

\subsection{Proof of Theorem ~\ref{th:critical_curve_triangular_lattice}}

This is an application of Theorem~1.1 in \cite{cimasoni2013critical} holding for a finite planar, non degenerate and doubly periodic weighted graph $G= (V, E)$.
Denote by $\cE(G)$ the set of all even subgraphs of $G$, that is, those subgraphs where the degree of each vertex is even. 
Further call ${\cE}_0(G)$ the set of even subgraphs of the lattice winding an even number of times around each direction of the torus and ${\cE}_1(G)={\mathcal E}(G)\setminus{\cE}_0(G)$.
Then the critical curve relating the parameters $J$ and $q$ of the Hamiltonian is the solution of the equation

\begin{equation}
	\sum_{\gamma\in{\cE}_0(G)}\prod_{e\in\gamma}\tanh J_e=\sum_{\gamma\in{\cE}_1(G)}\prod_{e\in\gamma}\tanh J_e.
\label{eq:graphs_critical_curve}
\end{equation}

The square lattice induced by the shaken dynamics on the triangular lattice, with $J_e = q$ for the self--interaction edges and $J_e = J$ for the other edges,  satisfies the hypotheses of this theorem and can be obtained by periodically repeating the elementary cell of Figure~\ref{fig:elementary_trianguar_shaken}.

\begin{figure}
    \centering
    \subfigure[]
    {\includegraphics[width=0.22\linewidth]{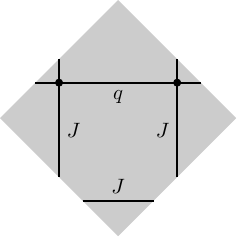}} \hfill
    \subfigure[]
    {\includegraphics[width=0.22\linewidth]{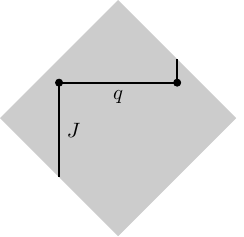}} \hfill
    \subfigure[]
    {\includegraphics[width=0.22\linewidth]{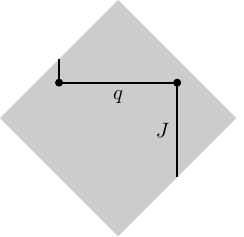}} \hfill
    \subfigure[]
    {\includegraphics[width=0.22\linewidth]{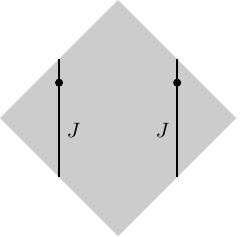}} \\
    \subfigure[]
    {\includegraphics[width=0.22\linewidth]{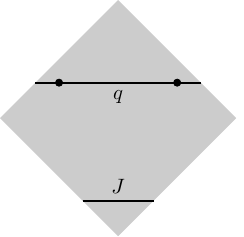}} \hfill
    \subfigure[]
    {\includegraphics[width=0.22\linewidth]{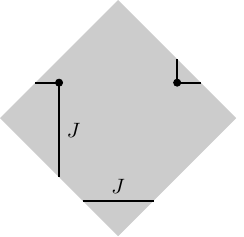}} \hfill
    \subfigure[]
    {\includegraphics[width=0.22\linewidth]{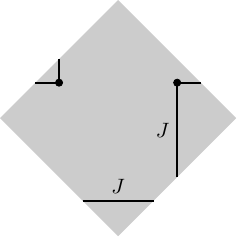}} \hfill
    \subfigure[]
    {\includegraphics[width=0.22\linewidth]{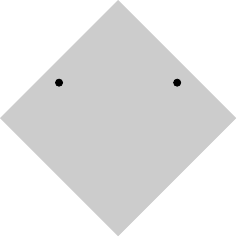}}
\caption{The elementary cell (a) for the shaken interaction on the triangular lattice and the corresponding even subgraphs. Subgraphs (a) and (h) wind around the torus an even number of times and are, therefore, in $\cE_0$ whereas the remaining subgraphs are in $\cE_1$.}
\label{fig:elementary_trianguar_shaken}
\end{figure}
A direct application of \eqref{eq:graphs_critical_curve} yields the claim.
\hfill $\qed$

\subsection{Proof of Theorem~\ref{corollary_q>}}
Let $H(\sigma)$ be a Hamiltonian  of the form given in (\ref{ham1}) and let $H(\sigma, \tau)$ be its doubling.

The invariant measure $\pi^b$ of the alternate dynamics $P^{alt}$ defined on the bipartite graph $G^b$ with Hamiltonian $H(\sigma, \tau)$ is identified in Theorem~\ref{t0}:
\begin{equation*}
{{\pi^b}}(\hbox{\boldmath $ \sigma$}):=\frac{1}{Z}e^{-\beta H(\hbox{{\boldmath $ \sigma$}})}.
\end{equation*}
At very low temperature, this measure concentrates on the set of configurations minimizing the Hamiltonian $H(\sigma, \tau)$. We have 
\begin{equation}\label{eq:minimizer_for_tau_equal_sigma}
   \min_{\sigma,\tau}H(\sigma,\tau) \leq \min_{\sigma} H({\sigma}, {\sigma})=\min_{\sigma} H(\sigma) - q |V|
\end{equation}
yielding immediately claim (2).

If the equality 
\begin{equation}\label{eq:minimizer_for_tau_equal_sigma_vecchio}
    \min_{\sigma,\tau}H(\sigma,\tau)=\min_{\sigma}H(\sigma,\sigma)
\end{equation}
holds, the parallel algorithm provided by the alternate dynamics may be used to find configurations minimizing $H(\sigma)$.

If the parameter $q$ satisfies condition \eqref{q>}, the validity of equation \eqref{eq:minimizer_for_tau_equal_sigma_vecchio} can be verified by contradiction. Assume indeed that there exists a pair configuration $(\bar{\sigma}, \bar{\tau})$ such that $$
\min_{\sigma,\tau}H(\sigma,\tau)= H(\bar{\sigma}, \bar{\tau})
$$
and $\bar{\sigma}\neq \bar{\tau}$ at least in a vertex $x \in V$.
Under condition \eqref{q>}, a spin flip at  vertex $x$ leads to a lower value for the doubling Hamiltonian, contradicting the hypothesis that the pair $(\bar{\sigma}, \bar{\tau})$ is the minimizer of 
$H(\sigma, \tau)$. 
\hfill $\qed$

\section{Conclusions and open problems}
\label{dop}
We conclude our paper with some  general comments and open problems.

With the shaken dynamics we have constructed a reversible parallel dynamics and we control its invariant measure
with arbitrary boundary conditions. The advantages of the shaken dynamics can be summarized as follows:
\begin{itemize}[label=\adfbullet{43}]
    \item The shaken prescription can be applied 
          to general interaction graphs. This allows to construct parallel algorithms to tackle a large class of optimization problems.
    \item The shaken prescription,
          modifying suitably the parameters appearing in the doubled Hamiltonian, allows to compare the spin systems defined on different geometries.
    \item The dynamics can be interpreted 
          as a model for systems in which some kind of interaction alternates its direction on short timescale. See below for an example referring to the tidal dissipation.
\end{itemize}

As noted by an anonymous referee, the shaken dynamics could be extended in order to consider spin systems with more general summable interactions, not necessarily limited to two or one body terms. This could have important applications in combinatorial optimization problems such as set covering problem.

The construction of the shaken dynamics
and, in particular,
of its generalization,
is not a unique prescription.
This freedom in the definition of the oriented graph
defining the dynamics
and in the choice of the
parameters involved could be usefully exploited
in applications to speed up the dynamics.

Finally we want to outline that the presence of an alternate interaction suggests that the shaken dynamics, with $B\not=\emptyset$ and $q$ large, could be a good model to take into account the effects of Earth's tides in geodynamics and other tidal dissipative phenomena in Solar System. We assume that the inner structure of the Earth and of the satellites of the major planets may be described in terms of constraints that can be randomly broken, with a probability depending on the state of the nearest neighbors of each constraint. Tidal effects could give a dependence of this breaking probability on an alternate direction, related to the tidal state and to the related tidal currents. This geological and astronomical application will be developed in forthcoming papers.

\noindent
{\bf Acknowledgments:}
The authors are grateful to Emilio Cirillo, Alexandre Gaudillière and Gabriella Pinzari for useful discussions.
We thank the anonymous referees who greatly helped to improve the
readability of the whole work.
B.S. and E.S. thank
the support of the A*MIDEX project (n. ANR-11-IDEX-0001-02) funded by the  ``Investissements d'Avenir" French Government program, managed by the French National Research Agency (ANR).
B.S. acknowledges the MIUR Excellence Department Project awarded to
the Department of Mathematics, University of Rome Tor Vergata, CUP E83C18000100006. E.S. has been supported  by the PRIN 20155PAWZB ``Large Scale Random Structures''. A.T. has been supported by Project FARE 2016 Grant R16TZYMEHN.

\bibliographystyle{amsplain}

\begin{thebibliography}{10}

\bibitem{apollonio2019criticality}
Valentina Apollonio, Roberto D'Autilia, Benedetto Scoppola, Elisabetta
  Scoppola, and Alessio Troiani, \emph{Criticality of measures on 2-d ising
  configurations: from square to hexagonal graphs}, Journal of Statistical
  Physics \textbf{177} (2019), no.~5, 1009--1021.

\bibitem{bovier2016metastability}
Anton Bovier and Frank Den~Hollander, \emph{Metastability: a
  potential-theoretic approach}, vol. 351, Springer, 2016.

\bibitem{cimasoni2013critical}
David Cimasoni and Hugo Duminil-Copin, \emph{The critical temperature for the
  ising model on planar doubly periodic graphs}, Electronic Journal of
  Probability \textbf{18} (2013), 1--18.

\bibitem{cirillo2002note}
Emilio~NM Cirillo, \emph{A note on the metastability of the ising model: the
  alternate updating case}, Journal of statistical physics \textbf{106} (2002),
  no.~1, 385--390.

\bibitem{cirillo2008metastability}
Emilio~NM Cirillo, Francesca~R Nardi, and Cristian Spitoni, \emph{Metastability
  for reversible probabilistic cellular automata with self-interaction},
  Journal of Statistical Physics \textbf{132} (2008), no.~3, 431--471.

\bibitem{dss1}
Paolo Dai~Pra, Benedetto Scoppola, and Elisabetta Scoppola, \emph{Sampling from
  a gibbs measure with pair interaction by means of {PCA}}, Journal of
  Statistical Physics \textbf{149} (2012), no.~4, 722--737.

\bibitem{dss2}
\bysame, \emph{Fast mixing for the low temperature 2d ising model through
  irreversible parallel dynamics}, Journal of Statistical Physics \textbf{159}
  (2015), no.~1, 1--20.

\bibitem{d2021parallel}
Roberto D’Autilia, Louis~Nantenaina Andrianaivo, and Alessio Troiani,
  \emph{Parallel simulation of two dimensional ising models using probabilistic
  cellular automata}, Journal of Statistical Physics \textbf{184} (2021),
  no.~1, 1--22.

\bibitem{KJZ}
Marcus Kaiser, Robert~L Jack, and Johannes Zimmer, \emph{Acceleration of
  convergence to equilibrium in markov chains by breaking detailed balance},
  Journal of statistical physics \textbf{168} (2017), no.~2, 259--287.

\bibitem{UBQP}
Gary Kochenberger, Jin-Kao Hao, Fred Glover, Mark Lewis, Zhipeng L{\"u}, Haibo
  Wang, and Yang Wang, \emph{The unconstrained binary quadratic programming
  problem: a survey}, Journal of combinatorial optimization \textbf{28} (2014),
  no.~1, 58--81.

\bibitem{KV}
O~Kozlov and N~Vasilyev, \emph{Reversible markov chains with local interaction,
  multicomponent random systems}, Adv. Probab. Related Topics (1980), no.~6,
  451--469.

\bibitem{LS}
Carlo Lancia and Benedetto Scoppola, \emph{Equilibrium and non-equilibrium
  ising models by means of {PCA}}, Journal of Statistical Physics \textbf{153}
  (2013), no.~4, 641--653.

\bibitem{LSchap}
\bysame, \emph{Ising model on the torus and {PCA} dynamics: Reversibility,
  irreversibility, and fast tunneling}, Probabilistic Cellular Automata.
  Emergence, Complexity and Computation (Pierre~Yves Louis and Francesca~Romana
  Nardi, eds.), vol.~27, Springer, 2018, pp.~89--104.

\bibitem{MNOS}
Francesco Manzo, Francesca~R Nardi, Enzo Olivieri, and Elisabetta Scoppola,
  \emph{On the essential features of metastability: tunnelling time and
  critical configurations}, Journal of Statistical Physics \textbf{115} (2004),
  no.~1, 591--642.

\bibitem{OV}
Enzo Olivieri and Maria~Eul{\'a}lia Vares, \emph{Large deviations and
  metastability}, no. 100, Cambridge University Press, 2005.

\bibitem{pss1}
Aldo Procacci, Benedetto Scoppola, and Elisabetta Scoppola, \emph{Probabilistic
  cellular automata for low-temperature 2-d ising model}, Journal of
  Statistical Physics \textbf{165} (2016), no.~6, 991--1005.

\bibitem{pss2}
\bysame, \emph{Effects of boundary conditions on irreversible dynamics},
  Annales Henri Poincar{\'e} \textbf{19} (2018), no.~2, 443--462.

\end{thebibliography}
\providecommand{\bysame}{\leavevmode\hbox to3em{\hrulefill}\thinspace}
\providecommand{\MR}{\relax\ifhmode\unskip\space\fi MR }
\providecommand{\MRhref}[2]{%
  \href{http://www.ams.org/mathscinet-getitem?mr=#1}{#2}
}
\providecommand{\href}[2]{#2}

\end{document}